\RequirePackage{fix-cm}
\documentclass{svjour3}                     
\smartqed  
\usepackage{graphicx}

\usepackage[numbers]{natbib}
\usepackage{tikz}
\usepackage{multirow}
\usepackage{threeparttable}
\usepackage{diagbox}
\usepackage{makecell}
\usepackage{subfigure}
\usepackage[ruled,linesnumbered,vlined]{algorithm2e}
\begin{document}

\title{Object Reachability via Swaps under Strict and Weak Preferences
}


\author{Sen Huang         \and
        Mingyu Xiao 
}


\institute{Sen Huang \at
University of Electronic Science and Technology of China \\
              \email{huangsen47@gmail.com}           
           \and
            Mingyu Xiao (The corresponding author) \at
           University of Electronic Science and Technology of China\\
           \email{myxiao@gmail.com}
}

\date{Received: date / Accepted: date}

\newtheorem{prop}{Proposition}
\newtheorem{ob}{Observation}
\maketitle

\begin{abstract}
  The \textsc{Housing Market} problem is a widely studied resource allocation problem. In this problem, each agent can only receive a single object and has preferences over all objects. Starting from an initial endowment, we want to reach a certain assignment via a sequence of rational trades. We first consider whether an object is reachable for a given agent under a social network, where a trade between two agents is allowed if they are neighbors in the network and no participant has a deficit from the trade.
  Assume that the preferences of the agents are strict (no tie among objects is allowed). This problem is polynomial-time solvable in a star-network and NP-complete in a tree-network. It is left as a challenging open problem whether the problem is polynomial-time solvable when the network is a path. We answer this open problem positively by giving a polynomial-time algorithm.
  Then we show that when the preferences of the agents are weak (ties among objects are allowed), the problem becomes NP-hard when the network is a path and can be solved in polynomial time when the network is a star. Besides, we consider the computational complexity of finding different optimal assignments for the problem in the special case where the network is a path or a star.


\end{abstract}

\section{Introduction}
Allocating indivisible objects to agents is an important problem in both computer science and economics.
A widely studied setting is that each agent can only receive one single object and each agent has preferences over objects.
The problem is called the \textsc{Assignment} problem~\citep{G1973Assignment,Wilson1977Assignment}, or the \textsc{House Allocation} problem~\citep{abdulkadirouglu1998random,Manlove2013Algorithmics}.
When each agent is initially endowed with an object and we want to reallocate objects under some conditions without any monetary transfers, the problem
is known as \textsc{Housing Market} problem~\citep{Shapley1974}.
\textsc{Housing Market} has several real-life applications such as allocation of housings~\citep{abdulkadirouglu1999house}, organ exchange~\citep{roth2004kidney}, and so on.
There are two different preference sets for agents. One is \emph{strict}, which is a full ordinal list of all objects, and the other one is
 \emph{weak}, where agents are allowed to be indifferent between objects. Both preference sets have been widely studied.
Under strict preferences, the celebrated \textit{Top Trading Cycle} rule~\citep{Shapley1974} has several important properties, e.g. strategy-proofness, and has been frequently used in mechanism design.
Some modifications of \textit{Top Trading Cycle} rule, called \textit{Top Trading Absorbing Sets} rule and \textit{Top Cycles} rule,
were introduced for weak preferences~\citep{alcalde2011exchange,jaramillo2012difference}.
More studies of \textsc{Housing Market} under the two preference sets from different aspects can be found in the literature ~\citep{jaramillo2012difference,aziz2012housing,saban2013house,ehlers2014top,sonoda2014two,ahmad2017essays}.

Some rules, such as \textit{Top Trading Cycle}, allow a single exchange to involve many agents.
However, exchanges of objects involving only two agents are basic, since they only require the minimum number of participants in the exchange and are very common in real life.

In some models, it is implicitly assumed that all agents have a tie with others such that they know each other, and can make some exchanges. However, some agents often do not know each other, and cannot exchange,
even if they can mutually get benefits. So 
\citet{DBLP:conf/ijcai/GourvesLW17} studied \textsc{Housing Market} where the agents are embedded in a social network to denote the ability to exchange objects between them.
In fact, recently it is a hot topic to study resource allocation problems over social networks and analyze the influence of networks.
\citet{Abebe2016Fair} and \citet{Bei2017Networked} introduced social networks of agents into the \textsc{Fair Division} problem of cake cutting. They defined (local) fairness concepts based on social networks and then compared them to the classic fairness notions and designed new protocols to find envy-free allocations in cake cutting.
\citet{DBLP:conf/atal/Bredereck0N18} and \citet{beynier2018local} also considered network-based \textsc{Fair Division} in allocating indivisible resources, where they assumed that objects were embedded in a network and each agent could only be assigned a connected component.

In this paper, we study \textsc{Housing Market} in a social network with simple trades between pairs of neighbors in the network. In this model, there are the same number of agents and objects and each agent is initially endowed with a single object.
Each agent has a preference list of all objects. The agents are embedded in a social network which determines their ability to exchange their objects. Two agents may swap their items under two conditions:
they are neighbors in the social network, and they find it mutually profitable (or no one will become worse under weak preferences).
Under this model, many problems have been studied, i.e, \textsc{Object Reachability} (OR), \textsc{Assignment Reachability} (AR) and \textsc{Pareto Efficiency} (PE), see the references given by~\citet{DBLP:conf/ijcai/GourvesLW17}.
\textsc{Object Reachability} is to determine whether an object is reachable for a given agent from the initial endowment via swaps.
\textsc{Assignment Reachability} is to determine whether a certain assignment of objects to all agents is reachable.
\textsc{Pareto Efficiency} is to find a Pareto optimal assignment within all the reachable assignments.

For \textsc{Assignment Reachability}, it is known that the problem can be solved in polynomial time in trees~\citep{DBLP:conf/ijcai/GourvesLW17} and cycles~\citep{DBLP:journals/corr/abs-2005-02218} and becomes NP-hard in general graphs~\citep{DBLP:conf/ijcai/GourvesLW17} and complete graphs~\citep{DBLP:journals/corr/abs-2005-02218}.
 In this paper, we mainly consider \textsc{Object Reachability} and \textsc{Pareto Efficiency} under different preferences.
\citet{DBLP:conf/atal/DamammeBCM15} proved that it is NP-hard to find a reachable assignment maximizing  the
minimum of the individual utilities, which can be modified to show the NP-hardness of \textsc{Object Reachability} in complete graphs.
\citet{DBLP:conf/ijcai/GourvesLW17} further showed that \textsc{Object Reachability}
is polynomial-time solvable under star-structures and NP-complete under tree-structures.
For the network being a path, they solved the special case where the given agent is an endpoint (a leaf) in the path and left it unsolved for the general path case. All the above results are under strict preferences.
In this paper, we will answer the open problem positively by giving a polynomial-time algorithm for \textsc{Object Reachability}
in a path under strict preferences.
Furthermore, we prove that \textsc{Weak Object Reachability} (\textsc{Object Reachability} under the weak preferences)
is NP-hard in a path and polynomial-time solvable in a star.
Very recently, another research group also obtained a polynomial-time algorithm for \textsc{Object Reachability}
in paths under strict preferences independently~\cite{DBLP:journals/corr/abs-1905-04219}. The main ideas of the two algorithms are similar: first characterizing some structural properties of the problem and then reducing the problem to the polynomial-time solvable \textsc{2-Sat} problem. Furthermore, the paper~\cite{DBLP:journals/corr/abs-1905-04219} also showed that  \textsc{Object Reachability}
in cycles under strict preferences is polynomial-time solvable.
For \textsc{Pareto Efficiency},
\citet{DBLP:conf/ijcai/GourvesLW17} proved that the problem under strict preferences is NP-hard in general graphs and polynomial-time solvable in a star or path.
 In this paper, we further consider \textsc{Weak Pareto Efficiency} (\textsc{Pareto Efficiency} under the weak preferences) and show that it is NP-hard in a path.
Our results and previous results are summarized in Table~\ref{table1}, where we use OR to denote \textsc{Object Reachability} and PE to denote \textsc{Pareto Efficiency} and the precise definitions of OR and PE are given in Section 2.3.
\begin{table}[htbp]
    \centering
    \caption{The computational complexity results}\label{table1}
    \begin{tabular}{c|c|c|c|c|c}
        \hline
        \multirow{2}{*}{Preferences} & \multirow{2}{*}{Problems} & \multicolumn{4}{c}{The network}\\
        \cline{3-6}
         &  & Stars & Paths & Trees & General Graphs\\
        \hline
        \multirow{3}{*}{Strict} & \multirow{2}{*}{OR} &  \multirow{2}{*}{P~\citep{DBLP:conf/ijcai/GourvesLW17}} & \textbf{P} & \multirow{2}{*}{NP-hard~\citep{DBLP:conf/ijcai/GourvesLW17}} & \multirow{2}{*}{NP-hard~\citep{DBLP:conf/ijcai/GourvesLW17}}\\
        & & &(Theorem~\ref{the_sor}) & &\\
        \cline{2-6}
        & PE & P~\citep{DBLP:conf/ijcai/GourvesLW17} & P~\citep{DBLP:conf/ijcai/GourvesLW17} & ? & NP-hard~\citep{DBLP:conf/ijcai/GourvesLW17}\\
        \hline
        \multirow{4}{*}{Weak} & \multirow{2}{*}{OR} & \textbf{P} & \textbf{NP-hard}  & \textbf{NP-hard}  & \textbf{NP-hard}\\
         &  & (Lemma~\ref{lem_wor}) & (Theorem~\ref{the_wor}) & (Theorem~\ref{the_wor}) & (Theorem~\ref{the_wor})\\
        \cline{2-6}
        & \multirow{2}{*}{PE} & \multirow{2}{*}{?} & \textbf{NP-hard}  & \textbf{NP-hard}  & \multirow{2}{*}{NP-hard~\citep{DBLP:conf/ijcai/GourvesLW17}}\\
         &  &  & (Theorem~\ref{the_pareto}) & (Theorem~\ref{the_pareto}) & \\
        \hline
    \end{tabular}
\end{table}

In Table~\ref{table1}, our results are marked as bold and the problems left unsolved are denoted by `?'.
One of the major results in this paper is the polynomial-time algorithm for \textsc{Object Reachability} under strict preferences in a path.
Although paths are rather simple graph structures, \textsc{Object Reachability}
in a path is not easy at all, as mentioned in \citep{DBLP:conf/ijcai/GourvesLW17} that ``Despite its apparent simplicity, \textsc{Reachable Object} (\textsc{Object Reachability})
in a path is a challenging open problem when no restriction
on the agent's location is made. We believe that
this case is at the frontier of tractability.''
Our polynomial-time algorithm involves several techniques and needs to call solvers for the \textsc{2-Sat} problem after characterizing some structural properties of the problem.
On the other hand, it is a little bit surprising that \textsc{Object Reachability} in paths becomes NP-hard under weak preferences. Note that under weak preferences, the number of feasible swaps in an instance may increase. This may increase the searching space of the problem dramatically and make the problem harder to search for a solution.
Furthermore, we also show that \textsc{Weak Object Reachability} in stars is polynomial-time solvable. Although the algorithm will use the idea of the algorithm for the problem under strict preferences~\citep{DBLP:conf/ijcai/GourvesLW17}, it involves new techniques.
In order to study the problems systematically and further understand the computational complexity of the problems, we also consider \textsc{Pareto Efficiency} and give several hardness results.

The following part of the paper is organized as follows.
Section~\ref{sec-back} provides some background.
Section~\ref{sec-spath} tackles the reachability of an object for an agent in a path under strict preferences.
Section~\ref{sec-wpath} studies \textsc{Weak Object Reachability} and \textsc{Weak Pareto Efficiency} in paths, and
Section~\ref{sec-star} studies \textsc{Weak Object Reachability} and \textsc{Weak Pareto Efficiency} in stars.
Section~\ref{sec-con} makes some concluding remarks.
Partial results of this paper were presented on the thirty-third AAAI conference on artificial intelligence (AAAI 2019) and appeared as~\cite{or_aaai}.
Due to the space limitation, the conference version~\cite{or_aaai} only presented Theorem~\ref{the_sor} and Theorem~\ref{the_wor} with proof sketches.
In this version, we gave the full proofs of the two theorems and further considered \textsc{Weak Object Reachability} and \textsc{Weak Pareto Efficiency} in paths and stars.
Section~\ref{sec-wpath} and Section~\ref{sec-star} are newly added.

\section{Background}\label{sec-back}
\subsection{Models}
There are a set $N=\{ 1,...,n\}$ of $n$ agents and a set $O=\{ o_1,...,o_n\}$ of $n$ objects.
An \emph{assignment} $\sigma$ is a mapping from $N$ to $O$, where $\sigma(i)$ is the object held by agent $i$ in $\sigma$.
We also use $\sigma ^T(o_i)$ to denote the agent who holds object $o_i$ in $\sigma$.
Each agent holds exactly one object at all time.
Initially, the agents are endowed with objects, and the initial endowment is denoted by $\sigma _0$.
We assume without loss of generality that $\sigma _0(i)=o_i$ for every agent $i$.

Each agent has a preference regarding all objects.
A \emph{strict preference} is expressed as a full linear order of all objects.
Agent $i$'s preference is denoted by $\succ _i$, and $o_a\succ _i o_b$ indicates the fact that agent $i$ prefers object $o_a$ to object $o_b$. The whole strict preference profile for all agents is represented by $\succ$.
For \emph{weak preferences}, an agent may be indifferent between two objects.
For two disjoint subsets of objects $S_1$ and $S_2$, we use $S_1\succ _i S_2$ to indicate that all objects in $S_1$ (resp., $S_2$)
are equivalent for agent $i$ and agent $i$ prefers any object in $S_1$ to any object in $S_2$.
We use $o_a\succeq_i o_b$ to denote that agent $i$ likes $o_a$ at least as the same as $o_b$. Two relations $o_a\succeq_i o_b$ and $o_b\succeq_i o_a$ together imply that $o_a$ and $o_b$ are equivalent for agent $i$, denoted by $o_b=_i o_a$.
We may use $\succeq$ to denote the whole weak preference profile for all agents.
 We may simply use $o_{a_1}\succeq_i o_{a_2} \succeq_i \dots \succeq_i o_{a_l}$ (resp., $o_{a_1}\succ_i o_{a_2} \succ_i \dots \succ_i o_{a_l}$) to denote that $o_{a_j}\succeq_i o_{a_{j+1}}$ (resp., $o_{a_j}\succ_i o_{a_{j+1}}$) holds simultaneously for all $1\leq j \leq l-1$.

 When we discuss the utility of an agent or the social welfare, we may also need to define a \emph{value function} of agents on objects. Let $f_i(o)$ denote the value of object $o$ for agent $i$.
 Then $f_i(o_a)= f_i(o_b)$ implies $o_{a}=_i o_{b}$ and $f_i(o_a)> f_i(o_b)$ implies $o_{a}\succ_i o_{b}$. Given an assignment $\sigma$, the \emph{utilitarian social welfare} is defined to be
 $\sum_{i\in N}f_i(\sigma(i))$.

Let $G=(N,E)$ be an undirected graph as the social network among agents, where the edges in $E$ capture the capability of communication and exchange between two agents in $N$. An instance of \textsc{Housing Market} is a tuple $(N,O,\succ, G, \sigma _0)$ or $(N,O,\succeq, G, \sigma _0)$ according to the preferences being strict or weak.

\subsection{Dynamics}
The approach we take in this paper is dynamic, and we focus on individually rational trades between two agents.
A trade is \emph{individually rational} if each participant receives an object at least as good as the one currently held, i.e., for two agents $i$ and $j$ and an assignment $\sigma$, the trade between agents $i$ and $j$ on $\sigma$ is individually rational if $\sigma(j)\succeq _i \sigma(i)$ and $\sigma(i)\succeq _j \sigma(j)$.
Notice that, under strict preferences, a trade is individually rational if each participant receives an object strictly better than the one currently held since there are no objects equivalent to any agent.

We require that every trade is performed between neighbors in the social network $G$.
Individually rational trades defined according to $G$ are called \emph{swaps}.
A swap is an exchange, where two participants have the capability to communicate.

If an assignment $\sigma'$  results from applying a swap in assignment $\sigma$, then the swap is uniquely decided by the two assignments and we may use $(\sigma, \sigma')$ to denote the swap.
We may also use a sequence of assignments $(\sigma _0,\sigma _1,$
$\sigma _2,\dots,\sigma _t)$ to represent a sequence of swaps, where $\sigma _i$ results from a swap from $\sigma _{i-1}$ for any $i\in \{ 1,...,t\}$.
An assignment $\sigma'$ is \emph{reachable} if there exists a sequence of swaps $(\sigma _0,\sigma _1,\sigma _2,...,\sigma _t)$ such that $\sigma _t=\sigma'$.
An object $o\in O$ is \emph{reachable} for an agent $i\in N$ if there is a sequence of swaps $(\sigma _0,\sigma _1,\sigma _2,...,\sigma _t)$ such that $\sigma _t(i)=o$.
An assignment $\sigma$ \textit{Pareto-dominates} an assignment $\sigma'$ if for all $i\in N$, $\sigma(i)\succeq_i\sigma'(i)$ and there is an agent $j\in N$ such that $\sigma(j)\succ_j\sigma'(j)$.
An assignment $\sigma$ is \textit{Pareto optimal} if $\sigma$ is not \textit{Pareto-dominated} by any other assignment.

\subsection{Problems}
We mainly consider two problems under our model. The first one is \textsc{Object Reachability}, which is to check whether an object is reachable for an agent from the initial endowment via swaps. Another one is \textsc{Pareto Efficiency}, which is to find a Pareto optimal assignment within all the reachable assignments.
\textsc{Object Reachability} is formally defined below.

\noindent\rule{\linewidth}{0.2mm}
\vspace{-1mm}
\textsc{Object Reachability} (OR)\\
\textbf{Instance:} $(N,O,\succ,G,\sigma _0)$, an agent $k\in N$, and an object $o_l\in O$.\\
\textbf{Question:} Is object $o_l$ reachable for agent $k$?\vspace{-0.2cm}\\
\rule{\linewidth}{0.2mm}

Note that in \textsc{Object Reachability}, the preferences are strict by default.
When the preferences are weak, we call the problem \textsc{Weak Object Reachability}.
When the social network is a path $P$, the corresponding problem is called \emph{\textsc{(Weak) Object Reachability} in paths}.
For \textsc{Object Reachability} in paths, an instance will be denoted by $I=(N,O,\succ,P,\sigma _0, k\in N, o_l\in O)$,
where we assume without loss of generality that $l<k$. For path structures, we always assume without loss of generality that the agents are listed as $1,2,\dots,n$ on a line from left to right with an edge between any two consecutive agents. Below is an example of \textsc{Object Reachability} in paths.

\begin{example}
There are four agents. The path structure, preference profile and a sequence of swaps are given in Figure~\ref{fig_example1}.
\begin{figure}[htbp]
    \centering
    \begin{tikzpicture}
        [scale = 0.85, line width = 0.5pt,solid/.style = {circle, draw, fill = black, minimum size = 0.3cm},empty/.style = {circle, draw, fill = white, minimum size = 0.5cm}]
        \node [] (A) at (1.2,4) {$P:$};
        \node [empty, label = center:$1$] (B) at (2,4) {};
        \node [empty, label = center:$2$] (C) at (3,4) {};
        \node [empty, label = center:$3$] (D) at (4,4) {};
        \node [empty, label = center:$4$] (E) at (5,4) {};
        \node [label = right:$1:o_2\succ$\fbox{$o_1$}$\succ o_3\succ o_4$] (F) at (6,4) {};

        \node [] (A1) at (1.2,3.5) {$\sigma _0:$};
        \node [] (B1) at (2,3.5) {$o_1$};
        \node [] (C1) at (3,3.5) {$o_2$};
        \node [] (D1) at (4,3.5) {$o_3$};
        \node [] (E1) at (5,3.5) {$o_4$};
        \node [label = right:$2:o_4\succ o_3\succ o_1\succ$\fbox{$o_2$}] (F1) at (6,3.5) {};

        \node [] (A2) at (1.2,3) {$\sigma _1:$};
        \node [] (B2) at (2,3) {$o_2$};
        \node [] (C2) at (3,3) {$o_1$};
        \node [] (D2) at (4,3) {$o_3$};
        \node [] (E2) at (5,3) {$o_4$};
        \node [label = right:$3:o_1\succ o_4\succ$\fbox{$o_3$}$\succ o_2$] (F2) at (6,3) {};

        \node [] (A3) at (1.2,2.5) {$\sigma _2:$};
        \node [] (B3) at (2,2.5) {$o_2$};
        \node [] (C3) at (3,2.5) {$o_3$};
        \node [] (D3) at (4,2.5) {$o_1$};
        \node [] (E3) at (5,2.5) {$o_4$};
        \node [label = right:$4:o_3\succ o_1\succ o_2\succ$\fbox{$o_4$}] (F3) at (6,2.5) {};

        \draw (B) -- (C);
        \draw (C) -- (D);
        \draw (D) -- (E);
        \draw[<->] (B1) -- (C1);
        \draw[<->] (C2) -- (D2);
    \end{tikzpicture}
    \caption{An example of four agents on a path} \label{fig_example1}
\end{figure}
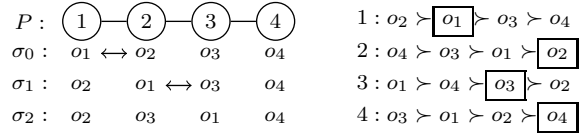

In Figure~\ref{fig_example1}, the initial endowments for agents are denoted by squares within the preferences.
After a swap between agents 1 and 2 from $\sigma _0$ we get $\sigma _1$ and after a swap between agents 2 and 3 from $\sigma _1$ we get $\sigma _2$.
Object $o_1$ is reachable for agent $3$.

\textsc{Pareto Efficiency} is formally defined below, where the preferences are strict by default.
When the preferences are weak, we call the problem \textsc{Weak Pareto Efficiency}.

\noindent\rule{\linewidth}{0.2mm}
\textsc{Pareto Efficiency} (PE)\\
\textbf{Instance:} $(N,O,\succ,G,\sigma _0)$.\\
\textbf{Question:} To find a Pareto optimal assignment within all the reachable assignments.\vspace{-0.1cm}\\
\rule{\linewidth}{0.2mm}
\vspace{-1mm}
\end{example}

\section{\textsc{Object Reachability} in Paths}\label{sec-spath}
In this section, the preferences are assumed to be strict.
\textsc{Object Reachability} is known to be NP-complete when the network is a tree and polynomial-time solvable when the network is a star~\citep{DBLP:conf/ijcai/GourvesLW17}. It is left unsolved whether the problem with the network being a path is NP-hard or not.
We show some properties of \textsc{Object Reachability} under path structures and design a polynomial-time algorithm for it.
In the remaining part of this section, the network is always assumed to be a path.

Recall that the problem is to check whether an object $o_l$ is reachable for an agent $k$ with $l< k$.
We introduce the main idea of the algorithm as follows.

\begin{enumerate}
    \item First, we show that the instance is equivalent to the instance after deleting all agents (and the corresponding endowed objects) on the left of agent $l$. Thus, we can assume the problem is to check whether object $o_1$ is reachable for an agent $k$.
    \item Second, we prove that if $o_1$ is reachable for agent $k$ then there is an object $o_{n'}$ with $n'\geq k$ that must be moved to agent $k-1$ in the final assignment. We show that we can also delete all agents and objects on the right of agent $n'$ in the initial endowment. However, we cannot find the agent $n'$ directly. In our algorithm, we will guess $n'$ by letting it be each possible value between $k$ and $n$. For each candidate $n'$, we get a special instance, called
        the neat $(o_1,o_{n'},k)$-\textsc{Constrained} instance, which contains $n'$ agents and objects and the goal is to check whether the first object $o_1$ is reachable for agent $k$ and the last object $o_{n'}$ is reachable for agent $k-1$ simultaneously, i.e., whether there is a reachable assignment $\sigma'$ such that $\sigma'(k)=o_1$ and $\sigma'(k-1)=o_{n'}$.
    \item Third, to solve the original instance, now we turn to solve at most $n$ neat $(o_1,o_{n'},k)$-\textsc{Constrained} instances.
    We characterize reachable assignments for neat $(o_1,o_{n'},k)$-\textsc{Constrained} instances by showing that any two objects must satisfy some \emph{compatible properties}.
\item Fourth, based on the compatible properties, we can prove that each object $o_i$ will be moved to either the left or the right side of its original position in a reachable assignment (i.e., can not stay at its original position), and there is at most one possible position $i_l$ for the left side and most one possible position $i_r$ for the right side. Furthermore, we can compute the candidate positions $i_l$ and $i_r$ in polynomial time.
 \item Last, we still need to decide whether to move each object $o_i$ to the left or the right side.
 (i.e., to decide which of $i_l$ and $i_r$ is the correct position for $o_i$ in the reachable assignment). Since there are at most two possible positions for each object, we can reduce the problem to the \textsc{2-Sat} problem and then solve it in polynomial time by using fast solvers for \textsc{2-Sat}.
\end{enumerate}

\subsection{Basic Properties and Neat Constrained Instances}

When the preferences are strict, we have the following observations and lemmas.

\begin{ob}\label{obs_1}
    Let $(\sigma _0,\sigma _1,\dots,\sigma _t)$ be a given sequence of swaps and
    $j\in N$ be an agent. For any $i\in\{0,1,\dots,t-1\}$, it holds that either $\sigma_{i+1}(j) = \sigma_i(j)$ or
    $\sigma_{i+1}(j) \succ_j \sigma_i(j)$.
\end{ob}

It implies the following lemma.

\begin{lemma}\label{lem_nondecreasing}
    Let $(\sigma _0,\sigma _1,\dots,\sigma _t)$ be a given sequence of swaps.
    For any two integers $i<j$ in $\{0,1,\dots,t\}$ and any agent $q\in N$, if $\sigma _i(q)=\sigma _j(q)$,
    then $\sigma _d(q)=\sigma _i(q)$ for any integer $i\leq d \leq j$.
\end{lemma}

Lemma~\ref{lem_nondecreasing} also says that an object will not `visit' an agent twice. This property is widely used in similar allocation problems under strict preferences.

Next, we analyze properties under the constraint that the social network is a path.
In a swap, an object is moved to the \emph{right side} if it is moved from agent $i$ to agent $i+1$, and an object is moved to the \emph{left side} if it is moved from agent $i$ to agent $i-1$.
In each swap, one object is moved to the right side and one object is moved to the left side.
We study the tracks of the objects in a feasible assignment sequence.

\begin{lemma}\label{lem_exactly}
For a sequence of swaps $(\sigma _0,\sigma _1,\dots,\sigma _t)$, if  $\sigma^T _t(o_i)=j$,
then there are exactly $|j-i|$ swaps including $o_i$. Furthermore, in all the $|j-i|$ swaps, $o_i$ is moved
to the right side  if $i<j$, and  moved to the left side  if $i>j$.
\end{lemma}

\begin{proof}
    First of all, we show that in a sequence of feasible swaps
    it is impossible for an object to be moved to the right and also to the left.
    Assume to the contrary that an object $o$ is moved in both directions and we will show a contradiction.
    Consider two consecutive swaps including $o$ (i.e., there are no swaps including $o$ between these two swaps) such that $o$ is moved in two opposite directions in the two swaps.
    Assume that object $o$ is moved away from agent $q$ after the first swap. We can see that $o$ will be moved back to agent $q$ after the second swap. Note that between these two swaps no swap includes $o$ and then $o$ will not be moved between them.
    This means object $o$ will visit agent $q$ twice, a contradiction to Lemma~\ref{lem_nondecreasing}. So any object can only move in one direction in a sequence of swaps.

    Since each swap including an object can only move the object to the neighbor position, we know that
    there are exactly $|j-i|$ swaps including $o_i$ that move $o_i$ to the right side if $i<j$ and to the left side if $i>j$. \qed
\end{proof}

\begin{lemma}\label{lem_track}
Let $(\sigma _0,\sigma _1,\dots,\sigma _t)$ be a sequence of swaps,
 and $o_a$ and $o_b$
be two objects with $a<b$. Let $a'=\sigma^T_t(o_a)$ and $b'=\sigma^T_t(o_b)$.
If $a' \leq a$  or  $b' \geq b$, then $a'< b'$.
\end{lemma}
\begin{proof}
    We consider three cases.
    If $a' \leq a$  and  $b' \geq b$, then by $a<b$ we get that $a' \leq a<b \leq b'$.
    Next, we assume that $a' \leq a$  and  $b' < b$. By Lemma~\ref{lem_exactly}, we know that both objects $o_a$ and $o_b$ can only be moved to the left side.
    If $a' \geq b'$, by $a<b$ we know that there must exist a swap including both $o_a$ and $o_b$. In this swap, two objects are moved in two opposite directions, a contradiction to the fact that the two objects can only be moved to the left side.
    The last case where $a' > a$  and  $b' \geq b$ can be proved similarly. \qed
\end{proof}

Lemma~\ref{lem_nondecreasing} shows that any object can only be moved in one direction.
Lemma~\ref{lem_track} shows that when an object is moved to the right side, all objects initially allocated on the left of it can
not be moved to the right of it at any time; when an object is moved to the left side, all objects initially allocated to the right of it can
not be moved to the left of it at any time.

In fact, if we want to move an object $o_l$ to an agent $k$ with $k>l$, we may not need to move any object to the left of $o_l$, i.e., objects $o_{l'}$ with $l'<l$. Equipped with Lemma~\ref{lem_track}, we can prove
\begin{lemma}\label{lem_simplifying1}
If object $o_l$ is reachable for agent $k$, then
there is a feasible assignment sequence $(\sigma _0,\sigma _1,\dots,\sigma _t)$ such that
$\sigma^T _t(o_l)=k$, and  $\sigma_t(i)=\sigma_0(i)$ for all $i < l$ if $l\leq k$ and for all $i > l$ if $l\geq k$.
\end{lemma}
\begin{proof}
    The result can be derived from Proposition 2 in~\citep{DBLP:conf/ijcai/GourvesLW17}.
    The authors give an algorithm to produce a sequence of swaps, which does not move any objects $o_i$ for all $i < l$ if $l\leq k$ and for all $i > l$ if $l\geq k$.
    Proposition 2 shows that if object  $o_l$ is reachable for agent $k$ then the obtained sequence of swaps is feasible. \qed
\end{proof}

For an instance $I=(N,O,\succ,P,\sigma _0, k, o_l)$ of \textsc{Object Reachability} in paths
with $l<k$, let $I'=(N',O',\succ', P',\sigma'_0, k, o_l)$ denote the instance obtained from $I$ by deleting agents $\{1,2,\dots, l-1\}$ and objects $\{o_1,o_2,\dots, o_{l-1}\}$.
In other words, we let $N'=\{l,l+1,\dots, n\}$, $O'=\{o_l,o_{l+1},\dots, o_n\}$,
and $\succ', P'$ and $\sigma'_0$ be the corresponding subsets of $\succ, P$ and $\sigma_0$.


\begin{corollary}
    \label{lem_simpl2}
Object $o_l$ is reachable for agent $k$ in the instance $I$ if and only if object $o_l$ is reachable for agent $k$ in the corresponding instance $I'$.
\end{corollary}

By Corollary~\ref{lem_simpl2}, we can always assume that the instance of \textsc{Object Reachability} in paths
is to check whether object $o_1$ is reachable for an agent $k$.

Assume that object $o_1$ is reachable for agent $k$. For a sequence of swaps $(\sigma _0,\sigma _1,\dots,\sigma _t)$
such that $\sigma _t(o_1)=k$, there are exactly $k-1$ swaps including $o_1$, where $o_1$ is moved to the right side according to
Lemma~\ref{lem_exactly}. The last swap including $o_1$ is between agent $k-1$ and agent $k$.
Let $o_{n'}$ denote the other object included in the last swap. In other words, after the last swap, agent $k-1$ will get object $o_{n'}$ and agent $k$ will get object $o_1$. Note that $o_{n'}$ is moved to the left side in this swap.
By Lemma~\ref{lem_exactly}, we know that object $o_{n'}$ is moved to the left side in all swaps including $o_{n'}$.
Therefore, we have the following observation.


\begin{ob}
    \label{ob_2}
    If there is a reachable assignment $\sigma$ such that $\sigma(k)=o_1$ and $\sigma(k-1)=o_{n'}$, then it holds that $n'\geq k$.
\end{ob}


\begin{definition}
   The \emph{$(o_1,o_{n'},k)$-\textsc{Constrained}} problem is to decide whether an instance $I$ of \textsc{Object Reachability} has a reachable assignment $\sigma$ such that $\sigma(k-1)=o_{n'}$ and $\sigma(k)=o_{1}$, where $n'\geq k$.
\end{definition}

Our idea is to transform our problem to the \emph{$(o_1,o_{n'},k)$-\textsc{Constrained}} problem. We do not know the value of
$n'$. So we search by letting $n'$ be each value in $\{k,k+1,\dots, n\}$. This will only increase the running time bound by a factor of $n$.
So we get

\begin{lemma}\label{lem6}
An instance $I=(N,O,\succ,P,\sigma _0, k, o_1)$ is a yes-instance
if and only if at least one of the $(o_1,o_{n'},k)$-\textsc{Constrained} instances for $n'\in \{k,k+1,\dots, n\}$ is a yes-instance.
\end{lemma}

For an $(o_1,o_{n'},k)$-\textsc{Constrained} instance $I$, we  use $I_{-{n'}}$ to denote the instance obtained from $I$ by deleting agents $\{n'+1,n'+2,\dots, n\}$ and objects $\{o_{n'+1},o_{n'+2},\dots, o_{n}\}$.

\begin{lemma} \label{lem_simpl3}
An $(o_1,o_{n'},k)$-\textsc{Constrained} instance $I$ is a yes-instance if and only if the corresponding instance $I_{-{n'}}$  is a yes-instance.
\end{lemma}
\begin{proof}
    We prove this lemma by induction on $n'$.
    It is clear that the lemma holds when $n'=2$. Assume that the lemma holds for $n'=n_0-1$. We show that the lemma also holds for $n'=n_0$.
    Let $n'=n_0$. When $I_{-{n'}}$  is a yes-instance, it is obvious that $I$ is also a yes-instance since we can use the same sequence of swaps as the solution to them.
    Next, we consider the converse direction and assume that $I$ is a yes-instance. Then there is a sequence of swaps $\{ \sigma_0,\sigma_1,\dots,\sigma_t\}$ such that $\sigma_t(k-1)=o_{n_0}$ and $\sigma_t(k)=o_1$.
    We consider the first swap $(\sigma_{r-1},\sigma_r)$ including object $o_{n_0}$. Since $o_{n_0}$ will be moved to agent $k-1$ with $k-1< n_0$, by Lemma~\ref{lem_exactly} we know that all swaps including $o_{n_0}$
    will move $o_{n_0}$ to the left side. So in the swap $(\sigma_{r-1},\sigma_r)$, object $o_{n_0}$ is also moved to the left side.
    By Lemma~\ref{lem_track}, we know that no object $o_i$ with $i>n_0$ is moved to the left of $o_{n_0}$ before $\sigma_r$.
    Let $\sigma_r'$ be the assignments of the first $n_0$ agents in $\sigma_r$. So
    we do not need to move any objects $o_i$ with $i>n_0$ to get $\sigma_r'$. Now object $o_{n_0}$ is at the position $n_0-1$.
    Since the lemma holds for $n'=n_0-1$, we know that there is a solution (a sequence of swaps) that does not move any objects $o_i$ with $i>n_0-1$ from $\sigma_r$ to a satisfying assignment.
    Such that there is a solution (a sequence of swaps) to $I$ that does not move any objects $o_i$ with $i>n_0$, which is also a solution to $I_{-{n'}}$. \qed
\end{proof}

By Lemma~\ref{lem_simpl3}, we can ignore all agents on the right side of $n'$ in an $(o_1,o_{n'},k)$-\textsc{Constrained} instance.

\begin{definition}
An $(o_1,o_{n'},k)$-\textsc{Constrained} instance is called \emph{neat} if $n'$ is the last agent.
\end{definition}

Next, we only consider neat $(o_1,o_{n'},k)$-\textsc{Constrained} instances.

\subsection{Characterization of Reachable Assignments}

To solve  neat $(o_1,o_{n'},k)$-\textsc{Constrained} instances,
we will give a full characterization of reachable assignments for them, which
shows that an assignment is reachable if and only if it is \emph{compatible} (defined in Definition~\ref{def_2}).
The compatibility reveals the relation between any two objects in a reachable assignment.

For two integers $x$ and $y$, we use $[x,y]$
to denote the set of integers between $x$ and $y$ (including $x$ and $y$). In the definition, we do not require  $x\leq y$.
So we have that $[x,y]$=$[y,x]$.

Consider a neat $(o_1,o_{n'},k)$-\textsc{Constrained} instance and an assignment $\sigma _t$ such that $\sigma^T_t(o_1)=k$ and $\sigma^T_t(o_{n'})=k-1$.
For any two objects $o_a$ and $o_b$, we let $a'=\sigma^T_t(o_a)$ and $b'=\sigma^T_t(o_b)$.

Note that since we only consider neat $(o_1,o_{n'},k)$-\textsc{Constrained} instances here, we know that all the objects will move before we get the final assignment where agent $k$ gets object $o_1$.

\begin{definition}\label{def_1}
    The two objects $o_a$ and $o_b$ are \emph{intersected} in assignment $\sigma _t$ if $Q=[a,a']\cap [b,b']$ is not an empty set.
\end{definition}

There are three kinds of intersections for two objects in an assignment. In the first two kinds of intersections, the two objects are moved in the same direction, i.e., $(a'-a)(b'-b)>0$.
In the last kind of intersections, the two objects are moved in the two opposite directions, i.e., $(a'-a)(b'-b)<0$.
The following Lemma~\ref{prop1} describes a property of the first two kinds of intersections and Lemma~\ref{prop2} describes a property of the third kind of intersections.

\begin{lemma}\label{prop1}
Let $(\sigma _0,\sigma _1,\dots,\sigma _t)$ be a sequence of swaps, and $o_a$ and $o_b$
be two objects with $1\leq a<b\leq n$. Let $a'=\sigma^T_t(o_a)$, $b'=\sigma^T_t(o_b)$, and $Q=[a,a']\cap [b,b']$. Assume that $Q\neq \emptyset$.\\
\textbf{(a)} If $a'>a$ and $b'>b$, then it holds that $o_a \succ_q o_b$ for all $q\in Q$;\\
\textbf{(b)} If $a'<a$ and $b'<b$, then it holds that $o_b \succ_q o_a$ for all $q\in Q$.
\end{lemma}
\begin{proof}
    Both objects  $o_a$ and $o_b$ will visit each agent in $Q$ during the sequence of swaps $(\sigma _0,\sigma _1,\dots,\sigma _t)$.

    \textbf{(a)} Since $b'>b$, by Lemma~\ref{lem_track}, we know that for each agent $q\in Q$ object $o_b$ will visit agent $q$ before object $o_a$.
    By Observation~\ref{obs_1}, we know that $o_a \succ_q o_b$.

    \textbf{(b)} Since $a'<a$, by Lemma~\ref{lem_track}, we know that for each agent $q\in Q$ object $o_a$ will visit agent $q$ before object $o_b$.
    By Observation~\ref{obs_1}, we know that $o_b \succ_q o_a$. \qed
\end{proof}

See Figure~\ref{figure_prop1} for an illustration of the two kinds of intersections in Lemma~\ref{prop1}.

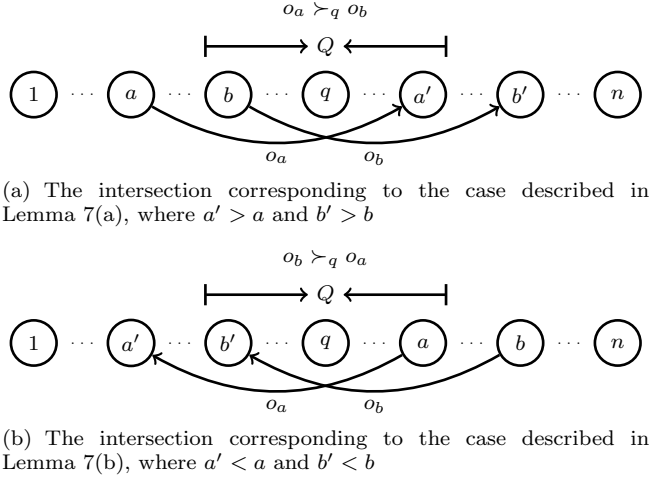
\begin{figure}[htbp]
    \centering
    \subfigure[The intersection corresponding to the case described in Lemma~\ref{prop1}(a), where $a'>a$ and $b'>b$]{
        \begin{tikzpicture}
            [scale = 0.64, line width = 1pt,solid/.style = {circle, draw, fill = black, minimum size = 0.3cm},empty/.style = {circle, draw, fill = white, minimum size = 0.6cm}]
            \node[empty] (A) at (1,1) {$1$};
            \node[] (B) at (2,1) {\tiny$\dots$};
            \node[empty] (C) at (3,1) {$a$};
            \node[] (D) at (4,1) {\tiny$\dots$};
            \node[empty] (E) at (5,1) {$b$};
            \node[] (F) at (6,1) {\tiny$\dots$};
            \node[empty] (G) at (7,1) {$q$};
            \node[] (H) at (8,1) {\tiny$\dots$};
            \node[empty, label = center:$a'$] (I) at (9,1) {};
            \node[] (J) at (10,1) {\tiny$\dots$};
            \node[empty, label = center:$b'$] (K) at (11,1) {};
            \node[] (L) at (12,1) {\tiny$\dots$};
            \node[empty] (M) at (13,1) {$n$};

            \node[label=above:$o_a\succ_q o_b$] (G1) at (7,2) {$Q$};

            \draw[|->] (4.5,2) -- (G1);
            \draw[|->] (9.5,2) -- (G1);
            \draw[->] (C) to[out = 330, in = 210] node[below]{$o_a$} (I);
            \draw[->] (E) to[out = 330, in = 210] node[below]{$o_b$} (K);
        \end{tikzpicture}
    }
    \hspace{1in}
    \subfigure[The intersection corresponding to the case described in Lemma~\ref{prop1}(b), where $a'<a$ and $b'<b$]{
        \begin{tikzpicture}
            [scale = 0.64, line width = 1pt,solid/.style = {circle, draw, fill = black, minimum size = 0.3cm},empty/.style = {circle, draw, fill = white, minimum size = 0.6cm}]
            \node[empty] (A) at (1,1) {$1$};
            \node[] (B) at (2,1) {\tiny$\dots$};
            \node[empty] (C) at (3,1) {$a'$};
            \node[] (D) at (4,1) {\tiny$\dots$};
            \node[empty] (E) at (5,1) {$b'$};
            \node[] (F) at (6,1) {\tiny$\dots$};
            \node[empty] (G) at (7,1) {$q$};
            \node[] (H) at (8,1) {\tiny$\dots$};
            \node[empty, label = center:$a$] (I) at (9,1) {};
            \node[] (J) at (10,1) {\tiny$\dots$};
            \node[empty, label = center:$b$] (K) at (11,1) {};
            \node[] (L) at (12,1) {\tiny$\dots$};
            \node[empty] (M) at (13,1) {$n$};

            \node[label=above:$o_b\succ_q o_a$] (G1) at (7,2) {$Q$};

            \draw[|->] (4.5,2) -- (G1);
            \draw[|->] (9.5,2) -- (G1);
            \draw[->] (I) to[out = 210, in = 330] node[below]{$o_a$} (C);
            \draw[->] (K) to[out = 210, in = 330] node[below]{$o_b$} (E);
        \end{tikzpicture}
    }

    \caption{An illustration of the two kinds of intersections in Lemma~\ref{prop1}}
    \label{figure_prop1}
\end{figure}

\begin{lemma}\label{prop2}
Let $(\sigma _0,\sigma _1,\dots,\sigma _t)$ be a sequence of swaps for a neat $(o_1,o_{n'},k)$-\textsc{Constrained} instance
such that $\sigma^T_t(o_1)=k$ and $\sigma^T_t(o_{n'})=k-1$, and $o_a$ and $o_b$
be two objects with $1\leq a<b\leq n'$. Let $a'=\sigma^T_t(o_a)$ and $b'=\sigma^T_t(o_b)$.
Assume that $a'>a$, $b'<b$ and $Q=[a,a']\cap [b,b']\neq \emptyset$. Let $Q'=[a+1,a']\cap [b,b']$.\\
\textbf{(a)} There is a swap including $o_a$ and $o_b$ which happens between agent $c-1$ and $c$, where $c=a'+b'-k+1\in Q'$; \\
\textbf{(b)} It holds that $o_b \succ_q o_a$ for all $\max(a,b')\leq q < c$, and $o_a \succ_q o_b$ for all $c \leq q \leq \min(a',b)$.
\end{lemma}
\begin{proof}
    Since $Q\neq \emptyset$, we know that there exists a swap including $o_a$ and $o_b$ in the sequence of swaps. We determine the position of this swap. We first consider \textbf{(a)}.

    By Lemma~\ref{lem_exactly}, we know that during the sequence of swaps each object will be moved in at most one direction.
    Let $o_{r_i}=\sigma_t(i)$ for all $1\leq i\leq n'$. Since $\sigma^T_t(o_1)=k$, by Lemma~\ref{lem_track} we know that $i<r_i$ for all $1\leq i<k$. Thus objects $o_{r_i}$ for all $1\leq i<k$ are moved
    to the left side. Since $\sigma^T_t(o_{n'})=k-1$, by Lemma~\ref{lem_track} we know that $i>r_i$ for all $k\leq i\leq n'$.
    Thus objects $o_{r_i}$ for all $k\leq i\leq n'$ are moved to the right side.

    Since $\sigma^T_t(o_1)=k$ and $o_a$ is moved to the right side, by Lemma~\ref{lem_track} we know that $a'\geq k$.
    By Lemma~\ref{lem_track} again, we know that   $1\leq r_i\leq a$   for each $k\leq i \leq a'$.
    For each object in $\{o_{a+1},\dots,o_b\}$, if it is moved to the right side then it can only by assigned to agents in $\{k, k+1,\dots, a'\}$ in $\sigma _t$ by Lemma~\ref{lem_track}.
    So there are exactly $a'-k+1$ objects in $\{o_{a+1},\dots,o_b\}$ that are moved to the right side. Furthermore, all of them are moved to the right of $k-1$ in $\sigma _t$.

    Since object $o_b$ is moved to the left side from $b$ to $b'$, by Lemma~\ref{lem_track} we know that there are exactly $b-b'$ objects in
    $\{o_1,\dots,o_b\}$ that will be moved to the right of $o_b$ in $\sigma _t$, which are all moved in the right direction.
    We further show that no object $o_x$ in $\{o_1,\dots,o_b\}$ moved in the right direction will stay at the left of $b'$ in $\sigma _t$.
    Since $\sigma^T_t(o_{n'})=k-1$ and $o_b$ is moved to the left side,
    by Lemma~\ref{lem_track} we know that $b'\leq k-1$, and all objects that are moved to the right side are at the right of $o_1$ in $\sigma_t$. If $o_x$ exists, then  $\sigma_t^T(o_x)>k>b'$, which means that $o_x$ is moved to the right of $o_b$ in $\sigma_t$, a contradiction.
    So there are exactly $b-b'$ objects in $\{ o_1,\dots,o_b\}$ that are moved to the right side.

    By the above two results, we know that the number of objects in $\{o_{a+1},\dots,o_b\}$ that are moved in the right direction is $b-b'-a'+k-1$.
    Since     $o_b$ should be swapped with $o_a$, we know that $o_b$ needs to be swapped with all objects in $\{o_{a+1},\dots,o_b\}$ that are moved in the right direction.
    So right before the swap including $o_a$ and $o_b$, object $o_b$ is at agent $b-(b-b'-a'+k-1)=a'+b'-k+1$. Thus, $c=a'+b'-k+1$.
    We know that \textbf{(a)} holds.

    Next, we consider \textbf{(b)}.
    For any agent $q$ such that $\max(a,b')\leq q < c$, both $o_a$ and $o_b$ will visit it. By Lemma~\ref{lem_track}, we know that $o_a$ will visit $q$ before $o_b$.
    By Observation~\ref{obs_1}, we know that $o_b \succ_q o_a$.
    For any agent $q$ such that $c \leq q \leq \min(a',b)$, both $o_a$ and $o_b$ will visit it. By Lemma~\ref{lem_track}, we know that $o_b$ will visit $q$ before $o_a$.
    By Observation~\ref{obs_1}, we know that $o_a \succ_q o_b$. Thus \textbf{(b)} holds.\qed
\end{proof}

See Figure~\ref{figure_prop2} for an illustration of the intersection in Lemma~\ref{prop2}.

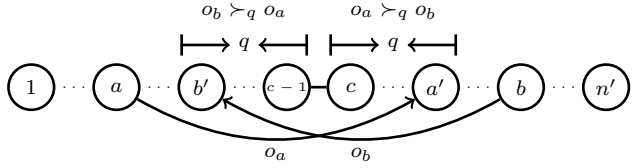
\begin{figure}[htbp]
    \centering
    \begin{tikzpicture}
        [scale = 0.56, line width = 1pt,solid/.style = {circle, draw, fill = black, minimum size = 0.3cm},empty/.style = {circle, draw, fill = white, minimum size = 0.6cm}]
        \node[empty] (A) at (1,1) {$1$};
        \node[] (B) at (2,1) {\tiny$\dots$};
        \node[empty] (C) at (3,1) {$a$};
        \node[] (D) at (4,1) {\tiny$\dots$};
        \node[empty, label = center:$b'$] (E) at (5,1) {};
        \node[] (F) at (6,1) {\tiny$\dots$};
        \node[empty, label = center:\tiny$c-1$] (G0) at (7,1) {};
        \node[empty] (G00) at (8.5,1) {$c$};
        \node[] (H) at (9.5,1) {\tiny$\dots$};
        \node[empty, label = center:$a'$] (I) at (10.5,1) {};
        \node[] (J) at (11.5,1) {\tiny$\dots$};
        \node[empty] (K) at (12.5,1) {$b$};
        \node[] (L) at (13.5,1) {\tiny$\dots$};
        \node[empty,label= center:$n'$] (M) at (14.5,1) {};

        \node[label=above:$o_b\succ_q o_a$] (F1) at (6,2) {$q$};
        \node[label=above:$o_a\succ_q o_b$] (H1) at (9.5,2) {$q$};

        \draw (G0) -- (G00);
        \draw[|->] (4.5,2) -- (F1);
        \draw[|->] (7.5,2) -- (F1);
        \draw[|->] (8,2) -- (H1);
        \draw[|->] (11,2) -- (H1);
        \draw[->] (C) to[out = 330, in = 210] node[below]{$o_a$} (I);
        \draw[->] (K) to[out = 210, in = 330] node[below]{$o_b$} (E);
    \end{tikzpicture}
    \caption{An illustration of the intersection in Lemma~\ref{prop2}, where $a'>a$ and $b'<b$}
    \label{figure_prop2}
\end{figure}

Based on Lemmas~\ref{prop1} and~\ref{prop2}, we give the conditions for assignment $\sigma _t$ to be a feasible assignment.

\begin{definition}\label{def_2}
    We say that a pair of objects $o_a$ and $o_b$ ($a< b$) are \emph{compatible} in assignment $\sigma_t$ if they are either not intersected or intersected and satisfying one of the following:\\
    \begin{enumerate}
        \item when $a'>a$ and $b'>b$, it holds that $a'<b'$ (corresponding to Lemma~\ref{lem_track}) and
        $o_a \succ_q o_b$ for all agents $q\in Q$ (corresponding to  Lemma~\ref{prop1}(a));\\
        \item when $a'<a$ and $b'<b$, it holds that $a'<b'$ (corresponding to Lemma~\ref{lem_track}) and  $o_b \succ_q o_a$ for all agents $q\in Q$ (corresponding to Lemma~\ref{prop1}(b));\\
        \item when $a'>a$ and $b'<b$, it holds that $c=a'+b'-k+1\in Q'=Q\setminus \{a\}$, $o_b \succ_q o_a$ for all $\max(a,b')\leq q < c$, and $o_a \succ_q o_b$ for all $c \leq q \leq \min(a',b)$ (corresponding to Lemma~\ref{prop2}).\\

    \end{enumerate}
\end{definition}

\begin{definition}\label{def_3}
An assignment $\sigma_t$ is \emph{compatible} if it holds that $\sigma_t(i)\neq o_i$ for any agent $i$ and any pair of objects in it are compatible.
\end{definition}

Lemma~\ref{prop1} and Lemma~\ref{prop2} imply that
\begin{lemma}\label{necessary}
If  $\sigma _t$ is a reachable assignment for a neat $(o_1,o_{n'},k)$-\textsc{Constrained} instance
such that $\sigma^T_t(o_1)=k$ and $\sigma^T_t(o_{n'})=k-1$, then $\sigma _t$ is compatible.
\end{lemma}
\begin{proof}
If there is an agent $i$ such that $i\leq k$ (resp., $i\geq k$) and $\sigma(i)= o_i$, then by Observation~\ref{obs_1} we know that agent $i$ does not participate in any swap
and then object $o_1$ cannot go to agent $k$ (resp., object $o_{n'}$ cannot go to agent $k-1$) since only two adjacent agents can participate in a swap.
So $\sigma(i)\neq o_i$ holds for any agent $i$.

For any pair of intersected objects $o_a$ and $o_b$, there are only three cases that are corresponding to the three items in Definition~\ref{def_2}.
By  Lemma~\ref{prop1} and Lemma~\ref{prop2}, we know that they must be compatible. \qed
\end{proof}

\begin{lemma}\label{oneswap}
Let $\sigma _t$ be a compatible assignment for a neat $(o_1,o_{n'},k)$-\textsc{Constrained} instance
such that $\sigma^T_t(o_1)=k$ and $\sigma^T_t(o_{n'})=k-1$.
For any two objects $o_{x-1}$ and $o_x$, if $\sigma^T_t(o_{x-1})> x-1$ and $\sigma^T_t(o_{x})< x$, then the swap between agents $x-1$ and $x$ in $\sigma_0$ is feasible. Let $\sigma_1$ denote the assignment after the swap between $x-1$ and $x$ in $\sigma_0$.
Then $\sigma _t$ is still compatible by taking $\sigma _1$ as the initial endowment.
\end{lemma}
\begin{proof}
    Consider the two objects $o_{x-1}$ and $o_x$, where $o_{x-1}$ is moved  in the right direction and $o_x$ is moved in the left direction. They are compatible.
    By item 3 in the definition of the compatibility, we know that $c=x$, which implies the swap between agents $x-1$ and $x$ is feasible in $\sigma_0$.

    Next, we show that if we take $\sigma _1$ as the initial endowment, the assignment $\sigma _t$ is still compatible.
    Compared with $\sigma_0$, the endowment positions of two objects are changed.
    To check whether $\sigma _t$ is compatible with $\sigma _1$ (being the initial endowment),
    we only need to check the compatibility of object pairs involving at least one of the two objects $o_{x-1}$
    and $o_x$.

    Now $\sigma_1(x-1)=o_{x}$, $\sigma_1(x)=o_{x-1}$ and $\sigma_1(i)=o_{i}$ for all other $i\notin\{ x-1,x\}$.

    First, we consider object pair $o_{x}$ and an object $o_i$ with $ i \geq x$. If $o_{x}$ and $o_i$ are intersected, then the
    intersection can only be the case described in Lemma~\ref{prop1}\textbf{(b)}.
    Item 2 in the definition of compatibility will hold because only the value of $a$ changes to a smaller value from $x$ to $x-1$ and the domain of $Q$ will not increase.

    Second, we consider object pair $o_{x}$ and an object $o_i$ with $1\leq i \leq x-2$.
    If $o_{x}$ and $o_i$ are intersected, there are two possible cases.
    When $\sigma_t(o_i)<i$, the
    intersection will be the case described in Lemma~\ref{prop1}\textbf{(b)}.
    Item 2 in the definition of compatibility will still hold because only the value of $b$ changes to a smaller value from $x$ to $x-1$
    and the domain of $Q$ will not increase.
    When $\sigma_t(o_i)>i$, the
    intersection will be the case described in Lemma~\ref{prop2}.
    We show that item 3 in the definition of compatibility will hold.
    The value of $c$ for $o_a=o_i$ and $o_b=o_x$ is the same no matter taking $\sigma_0$ or $\sigma_1$ as the initial endowment,
    since none of $a'$, $b'$ and $k$ is changed. Note that when taking $\sigma_0$ as the initial endowment  the value of $c$ is $x$ for $o_a=o_{x-1}$ and $o_b=o_x$. If the value of $c$ is also $x$ for $o_a=o_{i}$ and $o_b=o_x$ ($i\neq x-1$), then we will get a contradiction that both $o_{x-1}$ and $o_i$ will be assigned to the same agent in $\sigma_t$ by the computation formula of $c$.
    So we know that $c\neq x$ for $o_a=o_{i}$ and $o_b=o_x$. Thus, we get $c\in Q'=[i+1, \sigma _t^T(o_i)]\cap [\sigma _t^T(o_x) ,x-1]$ because  $c\in Q=[i+, \sigma _t^T(o_i)]\cap [\sigma _t^T(o_x),x]$ by $\sigma _t$ being compatible with $\sigma _0$ and $c\neq x$.
    After taking $\sigma_1$ as the initial endowment, the values of $a, a',c$ and $b'$ will not change, and the value of $b$ changes from $x$ to a smaller value $x-1$. We can see that the follows still hold:
    $o_b \succ_q o_a$ for all $\max(a,b')\leq q < c$, and $o_a \succ_q o_b$ for all $c \leq q \leq \min(a',b)$.

    Third, we consider $o_{x-1}$ and an object $o_i$ with $1\leq i \leq x-1$.
    If $o_{x-1}$ and $o_i$ are intersected, then the
    intersection can only be the case described in Lemma~\ref{prop1}\textbf{(a)}.
    Item 1 in the definition of compatibility will still hold because only the value of $b$ changes to a bigger value from $x-1$ to $x$ and the domain of $Q$ will not increase.

    Fourth, we consider $o_{x-1}$ and an object $o_i$ with $i \geq x$.
    If $o_{x-1}$ and $o_i$ are intersected, there are two possible cases.
    When $\sigma_t(o_i)>i$, the
    intersection will be the case described in Lemma~\ref{prop1}\textbf{(a)}.
    Item 1 in the definition of compatibility will still hold because only the value of $a$ changes to a bigger value from $x-1$ to $x$
    and the domain of $Q$ will not increase.
    When $\sigma_t(o_i)<i$, the
    intersection will be the case described in Lemma~\ref{prop2}. Analogously, we use similar arguments for the second case,
    we can prove that item 3 in the definition of compatibility holds.

    So if $\sigma _t$ is compatible by taking $\sigma _0$ as the initial endowment, then it is compatible by taking $\sigma _1$ as the initial endowment.\qed
\end{proof}

Based on Lemma~\ref{oneswap} we will prove the following lemma.

\begin{lemma}\label{sufficient}
Let $\sigma _t$ be an assignment for a neat $(o_1,o_{n'},k)$-\textsc{Constrained} instance
such that $\sigma^T_t(o_1)=k$ and $\sigma^T_t(o_{n'})=k-1$.
If $\sigma _t$ is compatible, then $\sigma _t$ is a reachable assignment.
\end{lemma}
\begin{proof}
We show that we can find a sequence of swaps from $\sigma_0$ to $\sigma_t$ for each compatible assignment $\sigma _t$.
In our algorithm, we first move object $\sigma _t(1)$ from its current position to agent 1, then move object $\sigma _t(2)$ to agent 2, and so on.
A formal description of the procedure is that:
For $i$ from 1 to $k-1$, move object $\sigma _t(i)$ from its current position to agent $i$ by a sequence of
    swaps including it (if the current position of $\sigma _t(i)$ is agent $j$, then it is a sequence of $|j-i|$ swaps).
To prove the correctness of the algorithm, we need to show that each swap in the algorithm is feasible and finally we can get the assignment $\sigma _t$.

     First, we show that the first loop of the algorithm can be executed legally, i.e., $\sigma _t(1)$ can be moved to agent 1 by a sequence of swaps. The first swap between agents $x-1$ and $x$ in $\sigma_0$ is feasible by Lemma~\ref{oneswap}.
    We use $\sigma _1$ to denote the assignment after this swap in $\sigma_0$. Then  $\sigma _t$ is compatible by taking $\sigma _{1}$ as the initial endowment by Lemma~\ref{oneswap}. By applying Lemma~\ref{oneswap} iteratively,
    we know that the swap between agents $x-1-i$ and $x-i$ in $\sigma_i$ is feasible (for $i>0$).
    We use
    $\sigma _{i+1}$ to denote the assignment after the swap in $\sigma_i$ and assume that object $o_{x}$ is assigned to agent $1$ in $\sigma _{t_1}$.
    By Lemma~\ref{oneswap}, we know that the sequence of swaps from $\sigma_0$ to $\sigma_{t_1}$ are feasible and
    $\sigma _t$ is compatible by taking $\sigma _{t_1}$ as the initial endowment.

    Next, we consider  $\sigma _{t_1}$ as the initial endowment. Agent $1$ has already gotten object $o_x$ and we can ignore it. After deleting agent 1 and object $o_x$ from the instance, we get a new $(o_1,o_{n'},k)$-\textsc{Constrained} instance.
    The second loop of the algorithm is indeed to move the corresponding object to the first agent in the new instance.
    So the correctness of the second loop directly follows from the above argument for the first loop.

    By iteratively applying the above arguments, we can prove that each loop of the algorithm can be executed legally.
    Therefore, all swaps in the algorithm are feasible.

    Let $\sigma _{t'}$ be the assignment returned by the algorithm. It holds that  $\sigma _{t'}(i)=\sigma _{t}(i)$ for $i\leq k-1$.
    For $i\geq k$, we show that  $\sigma _{t'}(i)=\sigma _{t}(i)$ still holds.  For any object $o_j=\sigma _{t'}(i)$ with $i\geq k$,
    objects $o_j$ and $o_{n'}$ are intersected and the intersection can only be the case described in Lemma~\ref{prop2}.
    By item 3 in the definition of compatibility, we know the value $c$ will not change no matter what the endowment position of $o_j$ is.
    Thus, $o_j$ and $o_{n'}$ will be swapped between two fixed agents and this is the last swap including $o_j$. Object $o_j$ will arrive at the position $a'=\sigma^T_{t}(o_j)$. Therefore, $\sigma _{t'}=\sigma _{t}$.\qed
\end{proof}

By Lemmas~\ref{necessary} and~\ref{sufficient}, to solve a neat $(o_1,o_{n'},k)$-\textsc{Constrained} instance, we only need to find a compatible assignment for it.

\subsection{Computing Compatible Assignments}
In a compatible assignment, object $o_1$ will be assigned to agent $k$ and object $o_{n'}$ will be assigned to agent $k-1$.
We consider other objects $o_i$ for $i\in\{2,3,\dots, n'-1\}$. In a compatible assignment, object $o_i$ will not be assigned to agent $i$ since each agent will participate in at least one swap including object $o_1$ or $o_{n'}$. There are two possible cases:  $o_i$ is assigned to agent $i'$ such that
$i'< i$; $o_i$ is assigned to agent $i'$ such that $i'> i$. We say that $o_i$ is moved to the left side for the former case and moved to the right side for the latter case. We will show that for each direction, there is at most one possible position for each object $o_i$ in a compatible assignment.

First, we consider $i\in\{2,3,\dots, k-1\}$. Assume that object $o_i$ is moved to the left side in a compatible assignment.
Thus, $o_1$ and $o_i$ are intersected and the intersection is of the case described in Lemma~\ref{prop2}.
We check whether there is an index $i'$ such that $i'\leq i$,  $o_i\succ_{i'-1}o_1$ and  $o_1\succ_j o_i$ for each $j \in \{i', i'+1, \dots, i\}$. The index $i'$ is also corresponding to the index $c$ in Lemma~\ref{prop2}.
We can see that $i'$ is the only possible agent for object $o_i$ to make $o_1$ and $o_i$ compatible if $o_i$ is moved to the left side.
We use $i_l$ to denote this agent $i'$ if it exists for $i$.
Assume that object $o_i$ is moved to the right side in a compatible assignment.
Since $o_1$ will be moved to agent $k$ and $o_i$ will be moved to the right side, by Lemma~\ref{lem_track} we know that $o_i$ will be moved to the right of $o_1$, i.e., an agent $i''$ with $i''> k$. Thus,  $o_i$ and $o_{n'}$ are intersected and the intersection is of the case described in Lemma~\ref{prop2}.
We check whether there is an index $i'$ such that $i'> k$,  $o_i\succ_{i'}o_{n'}$ and $o_{n'}\succ_j o_i$ for each $j \in \{k-1, k, \dots, i'-1\}$.
We can see that $i'$ is the only possible agent for object $o_i$ to make $o_{n'}$ and $o_i$ compatible if $o_i$ is moved to the right side.
We use $i_r$ to denote this agent $i'$ if it exists for $i$.

Second, we consider $i\in\{k,k+1,\dots, n'-1\}$. In fact, the structure of neat $(o_1,o_{n'},k)$-\textsc{Constrained} instances is symmetrical. We can rename the agents on the path from left to
right as $\{n', n'-1, \dots, 1\}$ instead of $\{1, 2, \dots, n'\}$ and then this case becomes the above case. We can compute $i_l$ and $i_r$ for each $i\in\{k,k+1,\dots, n'-1\}$ similarly.

The procedure to compute $i_l$ and $i_r$ for each agent $i\in\{2,3,\dots, n'-1\}$ in a neat $(o_1,o_{n'},k)$-\textsc{Constrained} instance
is presented as the following.

\begin{algorithm}[h!]
    \caption{To compute $i_l$ and $i_r$}
    \label{iril}
    \For{$2\leq i\leq n'-1$}
    {
      \If{there exists $i'$ such that $i'\leq i$,  $o_i\succ_{i'-1}o_1$ and  $o_1\succ_j o_i$ for each $j \in \{i', i'+1, \dots, i\}$}
      {let $i_l= i'$\;}
      \Else{let $i_l= \perp$ to indicate that $i_l$ does not exist.}
      \If{there exists $i'$ such that $i'> k$,  $o_i\succ_{i'}o_{n'}$ and $o_{n'}\succ_j o_i$ for each $j \in \{k-1, k, \dots, i'-1\}$}
      {let $i_r= i'$\;}
      \Else{let $i_r= \perp$ to indicate that $i_r$ does not exist.}
    }
\end{algorithm}

If none of $i_l$ and $i_r$ exists for some $i$, then this instance is a no-instance.
If only one of $i_l$ and $i_r$ exists, then object $o_i$ must be assigned to this agent in any compatible assignment.
The hardest case is that both $i_l$ and $i_r$ exist, where we may not know which agent the object will be assigned to in the compatible assignment.
In this case, we will rely on algorithms for \textsc{2-Sat} to find possible solutions.

For each agent $j\in\{1,2,3,\dots, n'\}$, we will use $R_j$ to store all possible objects that may be assigned to agent $j$ in a compatible assignment.
The following procedure computes the initial $R_j$.

\begin{enumerate}
\item  Initially, let $R_{k-1}=\{o_{n'}\}$, $R_{k}=\{o_1\}$, and $R_i=\emptyset$ for all other agent $i$.
\item  For each $i \in \{2,3,\dots,n'-1\}$, compute $i_l$ and $i_r$ by Algorithm~\ref{iril}, and add $o_i$ to $R_{i_l}$ if $i_l\neq \perp$ (resp., to $R_{i_r}$ if $i_r\neq \perp$).
\end{enumerate}

We also use the following two steps to iteratively update $R_j$ and then make the size of $R_j$ at most 2 for each $j$.

\begin{enumerate}
\item If there is an empty set $R_{j_0}$ for an agent $j_0$, stop and report the instance is a no-instance;
\item If there is a set $R_{j_0}$ containing only one object $o_{i_0}$ and the object $o_{i_0}$ appears in two sets $R_{j_0}$ and $R_{j'_0}$,
then delete $o_{i_0}$ from $R_{j'_0}$.
\end{enumerate}

The correctness of the second step is based on the fact that agent $j_0$ should get one object.
If there is only one candidate object $o_{i_0}$ for agent $j_0$,  then $o_{i_0}$ can only be assigned to agent $j_0$ in any compatible assignment.

We also analyze the running time of the above procedure to compute all $R_j$.
Algorithm~\ref{iril} computes $i_l$ and $i_r$ in $O(n)$ for each object $o_i$. Therefore, we use $O(n^2)$ time to set the initial values for all sets $R_j$. To update $R_i$, we may execute at most $n$ iterations and each iteration can be executed in $O(n)$. Hence, the procedure to compute all $R_j$ runs in time $O(n^2)$.

\begin{lemma}\label{lem_exact2}
    After the above procedure, either the instance is a no-instance or it holds that $1\leq|R_j|\leq 2$ for each $j\in\{1,2,\dots, n'\}$.
\end{lemma}
\begin{proof}
    We only need to consider the latter case where each set $R_j$ contains at least one object after the procedure.
    Note that for each set $R_{j}$ which contains only one object, the object will not appear in any other set.
    We ignore these singletons $R_{j}$ and consider the remaining sets.
    Each remaining set contains at least two objects. On the other hand, each other object $o_i$ (not appearing in a singleton) can be in at most two sets $R_{i_l}$ and $R_{i_r}$. On average each remaining set $R_{j}$ contains at most two objects.
    Therefore, each remaining set $R_{j}$ contains exactly two objects.\qed
\end{proof}

For a set $R_{j}$ containing only one object $o_i$, we know that object $o_i$ should be assigned to agent $j$ in any compatible
assignment. For sets $R_{j}$ containing two objects, we still need to decide which object is assigned to this agent such that we can get a compatible assignment.

We will reduce the remaining problem to \textsc{2-Sat}.
The \textsc{2-Sat} instance contains $n'$ variables $\{x_1, x_2, \dots, x_{n'}\}$ corresponding to the $n'$ objects.
When $x_i=1$, it means that object $o_i$ is moved to the right side, i.e, we will assign it to agent $i_r$ in the compatible assignment. When $x_i=0$, it means that object $o_i$ is moved to the left side and we will assign it to agent $i_l$ in the compatible assignment. In the \textsc{2-Sat} instance, we have two kinds of clauses, called \emph{agent clauses} and \emph{compatible clauses}.

For each set $R_{j}$, we associate $|R_{j}|$ literals with it. If there is an object $o_i$ such that $i_l=j$, we associate literal
$\overline{x_i}$ with $R_{j}$;  if there is an object $o_i$ such that $i_r=j$, we associate literal $x_i$ with $R_{j}$.
For each set $R_{j}$ of size 1 (let the associated literal be $\ell_j$), we construct one clause $c_j$ containing only one literal $\ell_j$.
For each set $R_{j}$ of size 2 (let the associated literals be $\ell_j^1$ and $\ell_j^2$), we construct two clauses $c_{j1}:\ell_j^1 \vee \ell_j^2$ and $c_{j2}:\overline{\ell_j^1} \vee \overline{\ell_j^2}$. These clauses are called \emph{agent clauses}.
From the construction, we can see that both of the two clauses $c_{j1}$ and $c_{j2}$ are satisfied if and only if exactly one object is assigned to agent $j$.
The agent clauses can guarantee that exactly one object is assigned to each agent.

For each pair of sets $R_{j}$ and $R_{i}$, we construct several clauses according to the definition of compatibility.
For any two objects $o_{j'}\in R_j$ (the corresponding literal associated with $R_j$ is $\ell_j$) and $o_{i'}\in R_i$ (the corresponding literal associated with $R_i$ is $\ell_i$), we say that $\ell_j$ and $\ell_i$ are \emph{compatible} if $o_{j'}$ and $o_{i'}$ are compatible when $o_{j'}$ is assigned to agent $j$ and $o_{i'}$ is assigned to agent $i$ in the assignment.
If $\ell_j$ and $\ell_i$ are not compatible, then either $o_{j'}$ cannot be assigned to agent $j$ or $o_{i'}$ cannot be assigned to agent $i$ in any compatible assignment. So we construct one \emph{compatible clause}: $\overline{\ell_j} \vee \overline{\ell_i}$ for each pair of
incompatible pair $\ell_j$ and $\ell_i$.
We can see that if the compatible clause $\overline{\ell_j} \vee \overline{\ell_i}$ is satisfied if and only if
either $o_{j'}$ is not assigned to agent $j$ or $o_{i'}$ is not assigned to agent $i$. Since we construct a compatible clause for any possible incompatible pair,
we know that all the compatible clauses are satisfied if and only if there is not incompatible pair in the assignment.
Note that, for each pair of sets $R_{j}$ and $R_{i}$, we will create at most $2\times2=4$ compatible clauses since each of $R_{j}$ and $R_{i}$ contains at most 2 elements.
Therefore, the number of compatible clauses in $O(n^2)$.

By Lemma~\ref{lem_exact2}, we know that there are at most two candidate objects for each agent in a compatible assignment.
Thus, each agent clause contains at most two literals.
By the construction of compatible clauses, we know that each compatible clause contains exactly two literals.
So, the constructed instance is a \textsc{2-Sat} instance.
The construction of the \textsc{2-Sat} instance implies the following lemma.

\begin{lemma}\label{reduce2sat}
The neat $(o_1,o_{n'},k)$-\textsc{Constrained} instance
has a compatible assignment if and only if the corresponding \textsc{2-Sat} instance constructed above has a satisfiable assignment.
\end{lemma}
\begin{proof}
First of all, we show that if the neat $(o_1,o_{n'},k)$-\textsc{Constrained} instance
has a compatible assignment then the \textsc{2-Sat} instance is satisfiable.
Let $\sigma$ be  a compatible assignment of the neat $(o_1,o_{n'},k)$-\textsc{Constrained} instance.
By the definition of compatible assignments, we know that each object is either moved to the left or the right in the path (cannot stay at its initial position) in the assignment $\sigma$. If object $o_i$ is moved to the left we let the corresponding
variable $x_i=0$ in the \textsc{2-Sat} instance, and  if object $o_i$ is moved to the right we let the corresponding
variable $x_i=1$. Thus, we assign values for all variables. Since each agent gets exactly one object in the compatible assignment $\sigma$, we know that all the agent clauses are satisfied.
Furthermore, $\sigma(i)$ and $\sigma(j)$ are compatible for any pair of agents $i$ and $j$ in the compatible assignment $\sigma$, and then all compatible clauses are satisfied.
So the \textsc{2-Sat} instance is satisfiable.

Next, we consider the other direction. Let $\Pi$ be a satisfying assignment of the \textsc{2-Sat} instance.
We construct an assignment $\sigma'$ of the neat $(o_1,o_{n'},k)$-\textsc{Constrained} instance:
if variable $x_i$ is assigned 1 in $\Pi$, let $\sigma'(i_r)=o_i$, and let  $\sigma'(i_l)=o_i$ otherwise.
We show that $\sigma'$ is a compatible assignment. First, each object is assigned to an agent different from its initial agent since each variable has been assigned to 1 or 0.
Second, each agent only gets one object, because all the agent clauses are satisfied. Last, any pair of objects are compatible, because all the compatible clauses are satisfied.
Such assignment $\sigma'$ is compatible.
 \qed
\end{proof}

\subsection{The Whole Algorithm}

The main steps of the whole algorithm to solve \textsc{Object Reachability} in paths are listed in Algorithm~\ref{alg:RO}. 

\begin{algorithm}[h]
    \caption{Main steps of the whole algorithm}
    \label{alg:RO}
    \KwIn{An instance $(N,O,\succ,P,\sigma,k\in N,o_1\in O)$}
    \KwOut{To determine whether $o_1$ is reachable for $k$}
    \For{each $n'\in \{k, k+1, \dots, n\}$}
    {
        Construct the corresponding neat $(o_1,o_{n'},k)$-\textsc{Constrained} instance by deleting agent $i$ and object $o_i$ for all $n'<i\leq n$\;
        Compute $i_r$ and $i_l$ for all $2\leq i\leq n'-1$ by Algorithm~\ref{iril}\;
        Compute and update $R_j$ for all  $1\leq j\leq n'$ by the procedure before Lemma~\ref{lem_exact2}\;
        Construct the corresponding \textsc{2-Sat} instance according to $R_j$ by the construction method introduced before Lemma~\ref{reduce2sat}\;
        \If{the \textsc{2-Sat} instance is a yes-instance}{ \textbf{return yes} and \textbf{quit}\;}
    }
    \textbf{return no}.
\end{algorithm}

The correctness of the algorithm follows from Corollary~\ref{lem_simpl2}, Lemma~\ref{lem6}, Lemma~\ref{lem_simpl3} and Lemma~\ref{reduce2sat}.
Corollary~\ref{lem_simpl2} says that we can simply reduce the problem to check whether $o_1$ is reachable for an agent. Then Lemma~\ref{lem6} and Lemma~\ref{lem_simpl3}
imply that the original instance is a yes-instance if and only if one of the corresponding neat $(o_1,o_{n'},k)$-\textsc{Constrained} instances is a yes-instance.
Lemma~\ref{reduce2sat} says that to solve a neat $(o_1,o_{n'},k)$-\textsc{Constrained} instance, we only need to solve the corresponding \textsc{2-Sat} instance.
So we solve all the $O(n)$ corresponding neat $(o_1,o_{n'},k)$-\textsc{Constrained} instances for $n'\in \{k, k+1, \dots, n\}$ by using \textsc{2-Sat} solvers, and
the algorithm returns \textbf{yes}
when one of the corresponding neat $(o_1,o_{n'},k)$-\textsc{Constrained} instances is a yes-instance.

Next, we analyze the running time bound of the algorithm.
The major part of the algorithm is to solve $O(n)$ neat $(o_1,o_{n'},k)$-\textsc{Constrained} instances.
To solve a neat $(o_1,o_{n'},k)$-\textsc{Constrained} instance, the above arguments show that we use $O(n^2)$ time to compute all sets $R_j$.
To construct the corresponding \textsc{2-Sat} instance, we construct at most $2n$ agent clauses in $O(n)$ time and construct at most $4{n \choose 2}$ compatible clauses, each of which will take $O(n)$ time to check the compatibility. So the \textsc{2-Sat} instance can be constructed in $O(n^3)$ time.
We use the $O(n+m)$-time algorithm for \textsc{2-Sat}~\citep{aspvall1982linear} to solve the instance, where $m$ is the number of clauses and in our instances, it holds that $m=O(n^2)$.
Thus, we can solve each neat $(o_1,o_{n'},k)$-\textsc{Constrained} instance in $O(n^3)$ time.
In total, the algorithm uses $O(n^4)$ time to compute all the $O(n)$ neat $(o_1,o_{n'},k)$-\textsc{Constrained} instances.

\begin{theorem}
  \label{the_sor}
\textsc{Object Reachability} in paths can be solved in $O(n^4)$ time.
\end{theorem}

We give an example to show the steps to compute a compatible assignment for a neat $(o_1,o_{n'},k)$-\textsc{Constrained} instance.

\begin{example}\label{exp_2}
    Consider a neat $(o_1,o_{n'},k)$-\textsc{Constrained} instance with $n'=8$ and $k = 5$ as shown in Figure 4.
   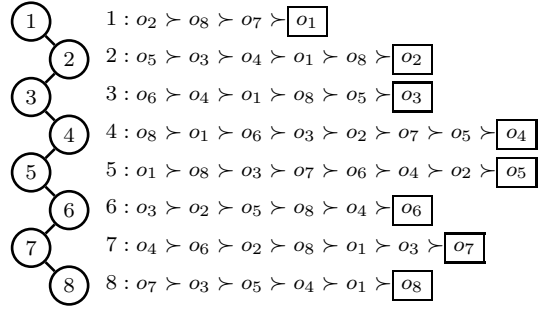
\begin{figure}[htbp]
        \centering
        \begin{tikzpicture}
        [scale =0.5, line width = 1pt,solid/.style = {circle, draw, fill = black, minimum size = 0.3cm},empty/.style = {circle, draw, fill = white, minimum size = 0.5cm}]
        \node[empty, label = center:$1$] (A) at (1,8) {};
        \node[empty, label = center:$2$] (B) at (2,7) {};
        \node[empty, label = center:$3$] (C) at (1,6) {};
        \node[empty, label = center:$4$] (D) at (2,5) {};
        \node[empty, label = center:$5$] (E) at (1,4) {};
        \node[empty, label = center:$6$] (F) at (2,3) {};
        \node[empty, label = center:$7$] (G) at (1,2) {};
        \node[empty, label = center:$8$] (H) at (2,1) {};

        \node[label = right:$1:o_2\succ o_8\succ o_7\succ$\fbox{$o_1$}] (A1) at (2.5,8) {};
        \node[label = right:$2:o_5\succ o_3\succ o_4\succ o_1\succ o_8\succ$\fbox{$o_2$}] (B1) at (2.5,7) {};
        \node[label = right:$3:o_6\succ o_4\succ o_1\succ o_8\succ o_5\succ$\fbox{$o_3$}] (C1) at (2.5,6) {};
        \node[label = right:$4:o_8\succ o_1\succ o_6\succ o_3\succ o_2\succ o_7\succ o_5\succ$\fbox{$o_4$}] (D1) at (2.5,5) {};
        \node[label = right:$5:o_1\succ o_8\succ o_3\succ o_7\succ o_6\succ o_4\succ o_2\succ$\fbox{$o_5$}] (E1) at (2.5,4) {};
        \node[label = right:$6:o_3\succ o_2\succ o_5\succ o_8\succ o_4\succ$\fbox{$o_6$}] (F1) at (2.5,3) {};
        \node[label = right:$7:o_4\succ o_6\succ o_2\succ o_8\succ o_1\succ o_3\succ$\fbox{$o_7$}] (G1) at (2.5,2) {};
        \node[label = right:$8:o_7\succ o_3\succ o_5\succ o_4\succ o_1\succ$\fbox{$o_8$}] (H1) at (2.5,1) {};

        \draw (A) -- (B);
        \draw (B) -- (C);
        \draw (C) -- (D);
        \draw (D) -- (E);
        \draw (E) -- (F);
        \draw (F) -- (G);
        \draw (G) -- (H);

        \end{tikzpicture}
        \caption{The instance of Example 2}        \label{fig-example2}
   \end{figure}

    We compute $i_r$ and $i_l$ for all $1\leq i\leq n'$ by Algorithm~\ref{iril}, the values of which are listed in Table~\ref{table_i}.
    \begin{table}[h!]
       \centering
       \caption{The values of $i_l$ and $i_r$}        \label{table_i}
       \begin{tabular*}{\linewidth}{@{\extracolsep{\fill}}ccccccccc}
            \hline
            Agent~$i$& $1$ & $2$ & $3$ & $4$ & $5$ & $6$ & $7$ & $8$ \\
            \hline
            $i_l$ & $\perp$ & $1$ & $2$ & $3$ & $2$ & $3$ & $\perp$ & $4$ \\
            $i_r$ & $5$ & $6$ & $6$ & $7$ & $6$ & $7$ & $8$ & $\perp$ \\
            \hline
        \end{tabular*}
    \end{table}

    According to the values in Table~\ref{table_i}, we compute and update $R_j$ for $j\in \{1,2,\dots, n'\}$ by the procedure before Lemma~\ref{lem_exact2}.
    After the update, it holds that $1\leq|R_j|\leq 2$ for all $1\leq j\leq n'$. The values of $R_j$ before and after updating are shown in the second and third columns of Table~\ref{table_RR}.

    \begin{table}[h!]
        \centering
        \caption{Initial and updated $R_j$ and agent clauses}\label{table_RR}
        \begin{tabular*}{\linewidth}{@{\extracolsep{\fill}}clll}
            \hline
            \it Sets & \it Initial $R_j$ & \it Updated $R_j$& \it Agent clauses \\
            \hline
            $R_1$ & $\{o_2\}$ & $\{o_2\}$ & $\overline{x_2}$ \\
            $R_2$ & $\{o_3,o_5\}$ & $\{o_3,o_5\}$ & $x_3\vee x_5,\overline{x_3}\vee\overline{x_5}$ \\
            $R_3$ & $\{o_4,o_6\}$ & $\{o_4,o_6\}$ & $x_4\vee x_6,\overline{x_4}\vee\overline{x_6}$ \\
            $R_4$ & $\{o_8\}$ & $\{o_8\}$ & $\overline{x_8}$ \\
            $R_5$ & $\{o_1\}$ & $\{o_1\}$ & $x_1$ \\
            $R_6$ & $\{o_2,o_3,o_5\}$ & $\{o_3,o_5\}$ & $x_3\vee x_5,\overline{x_3}\vee\overline{x_5}$ \\
            $R_7$ & $\{o_4,o_6\}$ & $\{o_4,o_6\}$ & $x_4\vee x_6,\overline{x_4}\vee\overline{x_6}$ \\
            $R_8$ & $\{o_7\}$ & $\{o_7\}$ & $x_7$ \\
            \hline
        \end{tabular*}
    \end{table}

    Next, we construct the corresponding \textsc{2-Sat} instance according to updated $R_j$.
    The \textsc{2-Sat} instance contains eight variables $x_1, x_2, \dots, x_8$ corresponding to the eight objects.
    There are two kinds of clauses: agents clauses and compatible clauses.

    Agent clauses are easy to construct: if $R_j$ contains only one object $o_q$, we construct one clause of a single literal $x_q$ if $q<j$ and one clause of a single literal $\overline{x_q}$ if $q>j$; if $R_j$ contains two objects $o_q$ and $o_p$, we construct two clauses of two literals according to our algorithm to constrain that exactly one
    object is assigned to agent $j$. The agent clauses for each set $R_j$ are shown in the last column of Table~\ref{table_RR}.

    We construct compatible clauses for all incompatible pairs. We check all pairs of objects and find that there
    are only two incompatible cases: $o_4$ and $o_5$ are incompatible if $o_4$ and $o_5$ are moved to agent $3$ and agent $2$, respectively; $o_4$ and $o_5$ are incompatible
    if $o_4$ and $o_5$ are moved to agent $7$ and agent $6$, respectively. So we construct the corresponding compatible clauses according to our algorithm and the compatible clauses are
    $$x_4\vee x_5 ~~~~\mbox{and}~~~~\overline{x_4}\vee\overline{x_5}.$$

    By using the $O(n+m)$ time algorithm for  \textsc{2-Sat} \citep{aspvall1982linear}, we get a feasible assignment $(1,0,0,0,1,1,1,0)$ for the \textsc{2-Sat} instance.
    The solution to the \textsc{2-Sat} instance indicates that objects $\{o_1,o_5,o_6,o_7\}$ are moved to the right side and the final positions for them are $\{5,6,7,8\}$;
    objects $\{o_2,o_3,o_4,o_8\}$ are moved to the left side and the final positions for them are $\{1,2,3,4\}$. The corresponding compatible assignment is $(o_2,o_3,o_4,o_8,o_1,o_5,o_6,o_7)$, i.e., $\sigma_7$ shown in Figure~\ref{fig_ans}.

    To obtain the corresponding sequence of feasible swaps from  $\sigma_0$ to $\sigma_7$, we use the algorithm in the proof of Lemma~\ref{sufficient}.
    We first move object $o_2$ to agent 1, then object $o_3$ to agent 2, then object $o_4$ to agent 3, and finally object $o_8$ to agent 4. The sequence of swaps is shown in Figure~\ref{fig_ans}.

    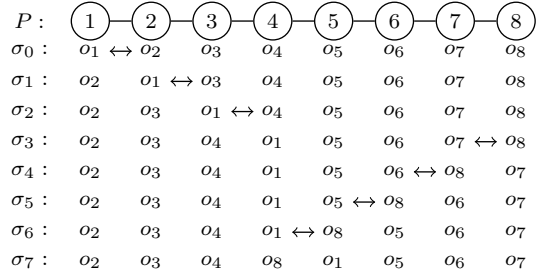
\begin{figure}[h!]
                \centering
            \begin{tikzpicture}
                [scale = 0.8, line width = 0.5pt,solid/.style = {circle, draw, fill = black, minimum size = 0.3cm},empty/.style = {circle, draw, fill = white, minimum size = 0.5cm}]
                \node [] (A) at (1,6) {$\sigma _0:$};
                \node [] (B) at (2,6) {$o_1$};
                \node [] (C) at (3,6) {$o_2$};
                \node [] (D) at (4,6) {$o_3$};
                \node [] (E) at (5,6) {$o_4$};
                \node [] (F) at (6,6) {$o_5$};
                \node [] (G) at (7,6) {$o_6$};
                \node [] (H) at (8,6) {$o_7$};
                \node [] (I) at (9,6) {$o_8$};

                \node [] (A1) at (1,5.5) {$\sigma _1:$};
                \node [] (B1) at (2,5.5) {$o_2$};
                \node [] (C1) at (3,5.5) {$o_1$};
                \node [] (D1) at (4,5.5) {$o_3$};
                \node [] (E1) at (5,5.5) {$o_4$};
                \node [] (F1) at (6,5.5) {$o_5$};
                \node [] (G1) at (7,5.5) {$o_6$};
                \node [] (H1) at (8,5.5) {$o_7$};
                \node [] (I1) at (9,5.5) {$o_8$};

                \node [] (A2) at (1,5) {$\sigma _2:$};
                \node [] (B2) at (2,5) {$o_2$};
                \node [] (C2) at (3,5) {$o_3$};
                \node [] (D2) at (4,5) {$o_1$};
                \node [] (E2) at (5,5) {$o_4$};
                \node [] (F2) at (6,5) {$o_5$};
                \node [] (G2) at (7,5) {$o_6$};
                \node [] (H2) at (8,5) {$o_7$};
                \node [] (I2) at (9,5) {$o_8$};

                \node [] (A3) at (1,4.5) {$\sigma _3:$};
                \node [] (B3) at (2,4.5) {$o_2$};
                \node [] (C3) at (3,4.5) {$o_3$};
                \node [] (D3) at (4,4.5) {$o_4$};
                \node [] (E3) at (5,4.5) {$o_1$};
                \node [] (F3) at (6,4.5) {$o_5$};
                \node [] (G3) at (7,4.5) {$o_6$};
                \node [] (H3) at (8,4.5) {$o_7$};
                \node [] (I3) at (9,4.5) {$o_8$};

                \node [] (A4) at (1,4) {$\sigma _4:$};
                \node [] (B4) at (2,4) {$o_2$};
                \node [] (C4) at (3,4) {$o_3$};
                \node [] (D4) at (4,4) {$o_4$};
                \node [] (E4) at (5,4) {$o_1$};
                \node [] (F4) at (6,4) {$o_5$};
                \node [] (G4) at (7,4) {$o_6$};
                \node [] (H4) at (8,4) {$o_8$};
                \node [] (I4) at (9,4) {$o_7$};

                \node [] (A5) at (1,3.5) {$\sigma _5:$};
                \node [] (B5) at (2,3.5) {$o_2$};
                \node [] (C5) at (3,3.5) {$o_3$};
                \node [] (D5) at (4,3.5) {$o_4$};
                \node [] (E5) at (5,3.5) {$o_1$};
                \node [] (F5) at (6,3.5) {$o_5$};
                \node [] (G5) at (7,3.5) {$o_8$};
                \node [] (H5) at (8,3.5) {$o_6$};
                \node [] (I5) at (9,3.5) {$o_7$};

                \node [] (A6) at (1,3) {$\sigma _6:$};
                \node [] (B6) at (2,3) {$o_2$};
                \node [] (C6) at (3,3) {$o_3$};
                \node [] (D6) at (4,3) {$o_4$};
                \node [] (E6) at (5,3) {$o_1$};
                \node [] (F6) at (6,3) {$o_8$};
                \node [] (G6) at (7,3) {$o_5$};
                \node [] (H6) at (8,3) {$o_6$};
                \node [] (I6) at (9,3) {$o_7$};

                \node [] (A7) at (1,2.5) {$\sigma _7:$};
                \node [] (B7) at (2,2.5) {$o_2$};
                \node [] (C7) at (3,2.5) {$o_3$};
                \node [] (D7) at (4,2.5) {$o_4$};
                \node [] (E7) at (5,2.5) {$o_8$};
                \node [] (F7) at (6,2.5) {$o_1$};
                \node [] (G7) at (7,2.5) {$o_5$};
                \node [] (H7) at (8,2.5) {$o_6$};
                \node [] (I7) at (9,2.5) {$o_7$};

                \node [label = center:$P:$] (A8) at (1,6.5) {};
                \node [empty,label = center:1] (B8) at (2,6.5) {};
                \node [empty,label = center:2] (C8) at (3,6.5) {};
                \node [empty,label = center:3] (D8) at (4,6.5) {};
                \node [empty,label = center:4] (E8) at (5,6.5) {};
                \node [empty,label = center:5] (F8) at (6,6.5) {};
                \node [empty,label = center:6] (G8) at (7,6.5) {};
                \node [empty,label = center:7] (H8) at (8,6.5) {};
                \node [empty,label = center:8] (I8) at (9,6.5) {};

                \draw (B8) -- (C8);
                \draw (C8) -- (D8);
                \draw (D8) -- (E8);
                \draw (E8) -- (F8);
                \draw (F8) -- (G8);
                \draw (G8) -- (H8);
                \draw (H8) -- (I8);

                \draw[<->] (B) -- (C);
                \draw[<->] (C1) -- (D1);
                \draw[<->] (D2) -- (E2);
                \draw[<->] (H3) -- (I3);
                \draw[<->] (G4) -- (H4);
                \draw[<->] (F5) -- (G5);
                \draw[<->] (E6) -- (F6);
            \end{tikzpicture}
            \caption{The sequence of swaps to reach the compatible assignment for Example~\ref{exp_2}} \label{fig_ans}
    \end{figure}
\end{example}

\textbf{Remark:} Although the algorithm will compute two possible values $i_l$ and $i_r$ for each $i$, it does not mean that object $o_i$ must be reachable for both $i_l$ and $i_r$.
In the above example, object $o_2$ is not  reachable for agent $2_r=6$ since agent 3 prefers its initially endowed object $o_3$ to $o_2$ and then $o_2$ cannot go to the right side.
In our algorithm, the compatible clauses can avoid assigning an object to an unreachable value $i_l$ or $i_r$. In the above example, if object $o_2$ goes to agent 6, then some object $o_i$ with $i>2$ will go to agent 1 or 2 and we will get an incompatible pair $o_2$ and $o_i$.

\section{Weak Preference Version in Paths}\label{sec-wpath}
We have proved that \textsc{Object Reachability} in paths is polynomial-time solvable. Next, we show that  \textsc{Weak Object Reachability} in paths is NP-hard.
One of the most important properties is that  Lemma~\ref{lem_nondecreasing} does not hold for  \textsc{Weak Object Reachability} and an object may `visit' an agent more than once.
Our proof involves a similar high-level idea as that of the NP-hardness proof for \textsc{Object Reachability} in a tree by~\citet{DBLP:conf/ijcai/GourvesLW17}.


The NP-hardness is proved by a reduction from the known NP-complete problem \textsc{2P1N-SAT}~\citep{yoshinaka2005higher}.
In a \textsc{2P1N-SAT} instance, we are given a set $V=\{ v_1,v_2,\dots,v_n\}$ of variables and a set $\mathcal{C}=\{C_1,C_2,\dots,C_m\}$ of clauses over $V$ such that every variable occurs 3 times in $\mathcal{C}$ with 2 positive literals and 1 negative literal.  The question is to check whether there is a variable assignment satisfying $\mathcal{C}$.
For a \textsc{2P1N-SAT} instance $I_{SAT}$, we construct an instance $I_{WOR}$ of \textsc{Weak Object Reachability} in paths such that $I_{SAT}$ is a yes-instance if and only if $I_{WOR}$ is a yes-instance.

Instance $I_{WOR}$ contains $6n+m+1$ agents and objects, which are constructed as follows.
There is an agent named $T$.
For each clause $C_i$ ($i\in\{1,\dots,m\}$), we introduce an agent also named $C_i$.
For each variable $v_i$, we add six agents, named as  $\overline{X}_i^{n_i}, {X}_i^{p_i},{X}_i^{q_i}, A_i^3, A_i^2$ and $A_i^1$. They form a path of length 5 in this order.
The path is called a \emph{block} and is denoted by $B_i$. See Figure~\ref{fig_bi}.
The names of the six agents have certain meaning: $\overline{X}_i^{n_i}$ means that the negative literal of $v_i$ appears in clause $C_{n_i}$;
${X}_i^{p_i}$ and ${X}_i^{q_i}$ mean that the positive literals of $v_i$ appear in the two clauses $C_{p_i}$ and $C_{q_i}$;  $A_i^3, A_i^2$ and $A_i^1$ are three auxiliary agents.

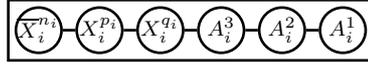
\begin{figure}[htbp]
        \centering
        \begin{tikzpicture}
            [scale =0.8, line width = 1pt,solid/.style = {circle, draw, fill = black, minimum size = 0.3cm},empty/.style = {circle, draw, fill = white, minimum size = 0.6cm}]
            \node [empty, label = center:\small$\overline{X}_i^{n_i}$] (A) at (1,1) {};
            \node [empty, label = center:$X_i^{p_i}$] (B) at (2,1) {};
            \node [empty, label = center:$X_i^{q_i}$] (C) at (3,1) {};
            \node [empty, label = center:$A_i^3$] (D) at (4,1) {};
            \node [empty, label = center:$A_i^2$] (E) at (5,1) {};
            \node [empty, label = center:$A_i^1$] (F) at (6,1) {};
            \draw (A) -- (B);
            \draw (B) -- (C);
            \draw (C) -- (D);
            \draw (D) -- (E);
            \draw (E) -- (F);
            \draw (0.5,0.5) rectangle (6.5,1.5);
        \end{tikzpicture}
        \caption{Block $B_i$} \label{fig_bi}
\end{figure}
The whole path is connected in the order shown in Figure~\ref{fig_path}.
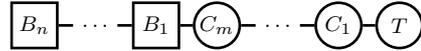
\begin{figure}[htbp]
    \centering
       \begin{tikzpicture}
        [scale = 0.8, line width = 1pt,solid/.style = {rectangle, draw, fill = white, minimum size = 0.6cm},empty/.style = {circle, draw, fill = white, minimum size = 0.6cm}]
        \node [solid, label = center:$B_n$,] (A) at (1,1) {};
        \node [] (B) at (2,1) {$\dots$};
        \node [solid, label = center:$B_1$] (C) at (3,1) {};
        \node [empty, label = center:$C_m$] (D) at (4,1) {};
        \node [] (E) at (5,1) {$\dots$};
        \node [empty, label = center:$C_1$] (F) at (6,1) {};
        \node [empty, label = center:$T$] (G) at (7,1) {};

        \draw (A) -- (B);
        \draw (B) -- (C);
        \draw (C) -- (D);
        \draw (D) -- (E);
        \draw (E) -- (F);
        \draw (F) -- (G);
    \end{tikzpicture}
    \caption{The structure of the whole path} \label{fig_path}
\end{figure}

In the initial assignment $\sigma_0$, object $t$ is assigned to agent $T$,
object $c_i$ is assigned to agent $C_i$ for $i\in \{1,2,\dots,m\}$,
object $a_i^j$ is assigned to agent $A_i^j$ for each $i\in \{1,2,\dots,n\}$ and $j\in \{1,2,3\}$, and
$\overline{o}_i^{n_i}$ (resp., ${o}_i^{p_i}$ and ${o}_i^{q_i}$) is assigned to agent $\overline{X}_i^{n_i}$ (resp., ${X}_i^{p_i}$ and ${X}_i^{q_i}$) for each $i\in \{1,2,\dots,n\}$.

Next, we define the preference profile $\succeq$. We only show the objects that each agent prefers at least as its initial one
and all other objects can be put behind its initial endowment in any order.
Let $L_i$ be the set of the objects associated with the literals in clause $C_i$. For example, $L_i$ is the object set $\{\overline{o}_a^{n_a}, {o}_b^{p_b}, {o}_c^{q_c}\}$ with $n_a=i$, $p_b=i$ and $q_c=i$.
The set $L_i$ will have the following property: in a reachable assignment, one object in $L_i$ will be assigned to agent $C_i$, which will be corresponding to the true literal in the clause $C_i$ in the
\textsc{2P1N-SAT} instance $I_{SAT}$.
For each variable $v_i$, we use $W_i$ to denote the set of objects $\{ c_1,\dots,c_m\}\cup\{ \overline{o}_j^{n_j}:j>i\}\cup\{ o_j^{p_j}:j\neq i\}\cup \{ o_j^{q_j}:j>i\}\cup\{ a_j^l:j<i,l=1,2,3\}$.
We are ready to give the preferences for the agents.

First, we consider the preferences for $T$ and $C_i$. The following preferences ensure that when
$C_i$ holds an object in $L_i$ for each $i\in \{1,\dots, m\}$, object $t$ is reachable for agent $C_m$ via a sequence of $m$ swaps between $C_i$ and $C_{i-1}$ for $i=1,\dots,m$, where $C_0=T$. Note that the squares in the preferences indicate the initial endowments for agents.\\

\noindent $T:L_1\succ$\fbox{$t$} ;\\
$C_i:L_{i+1}\succ t\succ L_i\succ c_1\succ L_{i-1}\succ \dots\succ c_{i-1}\succ L_1\succ$\fbox{$c_i$} , for all $1\leq i\leq m-1$ ;\\
$C_m:t\succ L_m\succ c_1\succ L_{m-1}\succ \dots\succ c_{m-1}\succ L_1\succ$\fbox{$c_m$} .\\

Next, we consider the preferences for the agents in each block $B_i$.
The following preferences ensure that at most one of $\overline{o}_i^{n_i}$ and $\{o_i^{p_i}, o_i^{q_i}\}$ can be moved to the right of the block,
which will indicate that the corresponding variable is set to 1 or 0. If $\overline{o}_i^{n_i}$ is moved to the right of the block, we will assign the corresponding variable 0;
if some of $\{o_i^{p_i}, o_i^{q_i}\}$ is moved to the right of the block, we will assign the corresponding variable 1.
We use the preference of $A_i^3$ to control this. Furthermore, we use $A_i^1$, $A_i^2$ and $A_i^3$ to (temporarily) hold  $\overline{o}_i^{n_i}$ (or $o_i^{p_i}$ and $o_i^{q_i}$) if they do not need
to be moved to the right of the block.\\

\noindent
$\overline{X}_i^{n_i}:W_i\cup\{ a_i^1,a_i^2,a_i^3,o_i^{p_i},o_i^{q_i}\}\succ$\fbox{$\overline{o}_i^{n_i}$} ,\\
$X_i^{q_i}:W_i\cup\{ a_i^1,a_i^2,a_i^3,\overline{o}_i^{n_i},o_i^{p_i}\}\succ$\fbox{$o_i^{q_i}$} ,\\
$X_i^{p_i}:W_i\cup\{ a_i^1,a_i^2,a_i^3,o_i^{q_i},\overline{o}_i^{n_i}\}\succ$\fbox{$o_i^{p_i}$} ,\\
$A_i^1:W_i\cup\{$\fbox{$a_i^1$}$,a_i^2,a_i^3,o_i^{p_i},o_i^{q_i},\overline{o}_i^{n_i}\}$ ,\\
$A_i^2:W_i\cup\{ a_i^1,$\fbox{$a_i^2$}$,a_i^3,o_i^{p_i},o_i^{q_i},\overline{o}_i^{n_i}\}$ ,\\
$A_i^3:W_i\cup\{ a_i^1,a_i^2,o_i^{p_i},o_i^{q_i}\}\succ\overline{o}_i^{n_i}\succ$\fbox{$a_i^3$} , for all $1\leq i\leq n$.\\

Instance $I_{WOR}$ is to determine whether object $t$ is reachable for agent $C_m$.

\begin{lemma}
    \label{lem_2p1n}
A \textsc{2P1N-SAT} instance $I_{SAT}$ is a yes-instance if and only if the corresponding instance $I_{WOR}$ of \textsc{Weak Object Reachability} in paths constructed above is a yes-instance.
\end{lemma}

\begin{proof}
First, we prove that a reachable assignment of $I_{WOR}$  implies a satisfying assignment for $I_{SAT}$.
     Assume that object $t$ is reachable for agent $C_m$. Then there are $m$ swaps including $t$ that happen between $C_i$ and $C_{i+1}$ for $i\in\{0,1,\dots, m-1\}$, where $C_0=T$. Note that the swap between $C_i$ and $C_{i+1}$ (where $C_i$ holds object $t$) can happen if and only if $C_{i+1}$ holds an object $a\in L_{i+1}$.
    We will prove the claim that for each block $B_i$, it is impossible that both $\overline{o}_i^{n_i}$ and one of ${o}_i^{p_i}$ and ${o}_i^{q_i}$ are moved to the right of this block.
    By this claim, we can see that a satisfying assignment for $I_{SAT}$ can be obtained by letting the literal corresponding to the object $a\in L_i$ to 1 for all agents $C_i$.

Next, we prove the claim that for each block $B_i$, it is impossible that both $\overline{o}_i^{n_i}$ and one of ${o}_i^{p_i}$ and ${o}_i^{q_i}$ are moved to the right of this block.
    Assume to the contrary that there is a block $B_i$ such that both $\overline{o}_n^{n_i}$ and one of ${o}_i^{p_i}$ and ${o}_i^{q_i}$ are moved to the right of this block.
    When $\overline{o}_k^{n_k}$ is moved to the right of this block, it must reach $A_k^3$ at some time.
    We know that $\overline{o}_i^{n_i}$ is the last but one in agent $A_i^3$'s given preferences. This implies that before $\overline{o}_i^{n_i}$ is swapped with $a_i^3$ between $X_i^{q_i}$ and $A_i^3$, $A_i^3$ cannot participate in any other swap.
    Therefore, the only way to move $\overline{o}_k^{n_k}$ to $A_i^3$ is that $\overline{o}_k^{n_k}$ is swapped with ${o}_i^{p_i}$, ${o}_i^{q_i}$ and $a_i^3$ in this order.
    For each agent of $\{\overline{X}_i^{n_i},X_i^{p_i},X_i^{q_i}\}$, once its endowment object is swapped with an object, the endowment object
    will not come back to the agent again.
    Thus, both $o_i^{p_i}$ and $o_i^{q_i}$ cannot be moved to the right of this block, which is a contradiction.

    On the other hand, if there is a satisfying assignment $\tau$ for $I_{SAT}$, we can construct a reachable assignment for $I_{WOR}$.
    We consider the value of each variable $v_i$ in $\tau$. If $v_i=1$, we move ${o}_i^{p_i}$ and ${o}_i^{q_i}$ to agents $A_i^3$ and $A_i^2$, respectively.
    The sequence of swaps is shown in Figure~\ref{casev1}. If $v_i=0$, we move $\overline{o}_i^{n_i}$ to agent $A_i^2$.
    The sequence of swaps is shown in Figure~\ref{casev0}. Note that for any case, all the swaps happen within a block.
    The objects ${o}_i^{p_i}$, ${o}_i^{q_i}$ and $\overline{o}_i^{n_i}$ in the above procedure are called \emph{true objects}.

    \begin{figure}[htbp]
        \centering
            \begin{tikzpicture}
                [line width = 1pt,solid/.style = {circle, draw, fill = black, minimum size = 0.3cm},empty/.style = {circle, draw, fill = white, minimum size = 0.6cm}]
                \node [empty, label = center:\tiny$\overline{X}_i^{n_i}$] (A) at (1,6) {};
                \node [empty, label = center:\tiny$X_i^{p_i}$] (B) at (2,6) {};
                \node [empty, label = center:\tiny$X_i^{q_i}$] (C) at (3,6) {};
                \node [empty, label = center:\small$A_i^3$] (D) at (4,6) {};
                \node [empty, label = center:\small$A_i^2$] (E) at (5,6) {};
                \node [empty, label = center:\small$A_i^1$] (F) at (6,6) {};

                \node [] (A1) at (1,5.5) {\small$\overline{x}_i^{n_i}$};
                \node [] (B1) at (2,5.5) {\small$x_i^{p_i}$};
                \node [] (C1) at (3,5.5) {\small$x_i^{q_i}$};
                \node [] (D1) at (4,5.5) {\small$a_i^3$};
                \node [] (E1) at (5,5.5) {\small$a_i^2$};
                \node [] (F1) at (6,5.5) {\small$a_i^1$};

                \node [] (A2) at (1,5) {\small$\overline{x}_i^{n_i}$};
                \node [] (B2) at (2,5) {\small$x_i^{p_i}$};
                \node [] (C2) at (3,5) {\small$a_i^3$};
                \node [] (D2) at (4,5) {\small$x_i^{q_i}$};
                \node [] (E2) at (5,5) {\small$a_i^2$};
                \node [] (F2) at (6,5) {\small$a_i^1$};

                \node [] (A3) at (1,4.5) {\small$\overline{x}_i^{n_i}$};
                \node [] (B3) at (2,4.5) {\small$x_i^{p_i}$};
                \node [] (C3) at (3,4.5) {\small$a_i^3$};
                \node [] (D3) at (4,4.5) {\small$a_i^2$};
                \node [] (E3) at (5,4.5) {\small$x_i^{q_i}$};
                \node [] (F3) at (6,4.5) {\small$a_i^1$};

                \node [] (A4) at (1,4) {\small$\overline{x}_i^{n_i}$};
                \node [] (B4) at (2,4) {\small$a_i^3$};
                \node [] (C4) at (3,4) {\small$x_i^{p_i}$};
                \node [] (D4) at (4,4) {\small$a_i^2$};
                \node [] (E4) at (5,4) {\small$x_i^{q_i}$};
                \node [] (F4) at (6,4) {\small$a_i^1$};

                \node [] (A5) at (1,3.5) {\small$\overline{x}_i^{n_i}$};
                \node [] (B5) at (2,3.5) {\small$a_i^3$};
                \node [] (C5) at (3,3.5) {\small$a_i^2$};
                \node [] (D5) at (4,3.5) {\small$x_i^{p_i}$};
                \node [] (E5) at (5,3.5) {\small$x_i^{q_i}$};
                \node [] (F5) at (6,3.5) {\small$a_i^1$};
                \draw (A) -- (B);
                \draw (B) -- (C);
                \draw (C) -- (D);
                \draw (D) -- (E);
                \draw (E) -- (F);
                \draw[<->] (C1) -- (D1);
                \draw[<->] (D2) -- (E2);
                \draw[<->] (B3) -- (C3);
                \draw[<->] (C4) -- (D4);
            \end{tikzpicture}
            \caption{The sequences of swaps to move ${o}_i^{p_i}$ and ${o}_i^{q_i}$ to agents $A_i^3$ and $A_i^2$, respectively}\label{casev1}
    \end{figure}
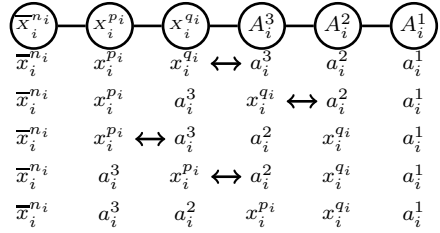

    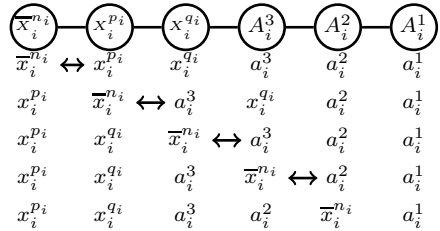
\begin{figure}[htbp]
        \centering
            \begin{tikzpicture}
                [line width = 1pt,solid/.style = {circle, draw, fill = black, minimum size = 0.3cm},empty/.style = {circle, draw, fill = white, minimum size = 0.6cm}]
                \node [empty, label = center:\tiny$\overline{X}_i^{n_i}$] (A) at (1,6) {};
                \node [empty, label = center:\tiny$X_i^{p_i}$] (B) at (2,6) {};
                \node [empty, label = center:\tiny$X_i^{q_i}$] (C) at (3,6) {};
                \node [empty, label = center:\small$A_i^3$] (D) at (4,6) {};
                \node [empty, label = center:\small$A_i^2$] (E) at (5,6) {};
                \node [empty, label = center:\small$A_i^1$] (F) at (6,6) {};

                \node [] (A1) at (1,5.5) {\small$\overline{x}_i^{n_i}$};
                \node [] (B1) at (2,5.5) {\small$x_i^{p_i}$};
                \node [] (C1) at (3,5.5) {\small$x_i^{q_i}$};
                \node [] (D1) at (4,5.5) {\small$a_i^3$};
                \node [] (E1) at (5,5.5) {\small$a_i^2$};
                \node [] (F1) at (6,5.5) {\small$a_i^1$};

                \node [] (A2) at (1,5) {\small$x_i^{p_i}$};
                \node [] (B2) at (2,5) {\small$\overline{x}_i^{n_i}$};
                \node [] (C2) at (3,5) {\small$a_i^3$};
                \node [] (D2) at (4,5) {\small$x_i^{q_i}$};
                \node [] (E2) at (5,5) {\small$a_i^2$};
                \node [] (F2) at (6,5) {\small$a_i^1$};

                \node [] (A3) at (1,4.5) {\small$x_i^{p_i}$};
                \node [] (B3) at (2,4.5) {\small$x_i^{q_i}$};
                \node [] (C3) at (3,4.5) {\small$\overline{x}_i^{n_i}$};
                \node [] (D3) at (4,4.5) {\small$a_i^3$};
                \node [] (E3) at (5,4.5) {\small$a_i^2$};
                \node [] (F3) at (6,4.5) {\small$a_i^1$};

                \node [] (A4) at (1,4) {\small$x_i^{p_i}$};
                \node [] (B4) at (2,4) {\small$x_i^{q_i}$};
                \node [] (C4) at (3,4) {\small$a_i^3$};
                \node [] (D4) at (4,4) {\small$\overline{x}_i^{n_i}$};
                \node [] (E4) at (5,4) {\small$a_i^2$};
                \node [] (F4) at (6,4) {\small$a_i^1$};

                \node [] (A5) at (1,3.5) {\small$x_i^{p_i}$};
                \node [] (B5) at (2,3.5) {\small$x_i^{q_i}$};
                \node [] (C5) at (3,3.5) {\small$a_i^3$};
                \node [] (D5) at (4,3.5) {\small$a_i^2$};
                \node [] (E5) at (5,3.5) {\small$\overline{x}_i^{n_i}$};
                \node [] (F5) at (6,3.5) {\small$a_i^1$};
                \draw (A) -- (B);
                \draw (B) -- (C);
                \draw (C) -- (D);
                \draw (D) -- (E);
                \draw (E) -- (F);

                \draw[<->] (A1) -- (B1);
                \draw[<->] (B2) -- (C2);
                \draw[<->] (C3) -- (D3);
                \draw[<->] (D4) -- (E4);

            \end{tikzpicture}
            \caption{The sequences of swaps to move $\overline{o}_i^{n_i}$ to agent $A_i^2$}\label{casev0}
    \end{figure}
    After this procedure, we move only one of true objects in $L_j$ to each agent $C_j$ one by one in an order where $C_j$ with smaller $j$ first gets its object in $L_j$ (for every agent in blocks, almost all objects in her given preferences are indifferent with each other).
    During this procedure, once another true object 
    is moved out of its position $A_i^3$ or $A_i^2$ (this may happen when the true object is on the moving path of another true object to $C_{j'}$ with $j<j'$), we will simply move it back to $A_i^3$ or $A_i^2$  by one swap. We can do this swap because almost all objects except $\overline{o}_i^{n_i}$ and $a_i^3$
    are equivalent for $A_i^3$, $A_i^2$ and $A_i^1$.
    So in each iteration, only one true object is moved out of its current position $A_i^3$ or $A_i^2$ and it is moved to its final position $C_j$ directly.
    Note that $C_{i+1}$ holds an object $a\in L_{i+1}$ for $i$ from $0$ to $m-1$ now and $C_0=T$ holds the object $t$. We make a swap between $C_i$ and $C_{i+1}$ for $i$ from $0$ to $m-1$, and then
    the object $t$ will be moved to agent $C_m$.
    \qed
\end{proof}

Lemma~\ref{lem_2p1n} implies that

\begin{theorem}
    \label{the_wor}
      \textsc{Weak Object Reachability} is NP-hard even when the network is a path.
\end{theorem}

The above reduction from \textsc{2P1N-SAT} can also be used to prove the NP-hardness of \textsc{Weak Pareto Efficiency} and the problem of maximizing the utilitarian social welfare.
We will use the following lemma.

\begin{lemma}\label{newadded1}
For the above instance $I_{WOR}$ constructed from the \textsc{2P1N-SAT} instance $I_{SAT}$, if $I_{WOR}$ is a yes-instance, then there is a reachable assignment $\sigma$ where every agent holds one of its most favorite objects.
\end{lemma}
\begin{proof}
    Let $\sigma$ be a reachable assignment in the above instance $I_{WOR}$ such that $\sigma (C_m)=t$.
    Clearly, each agent $C_i$ for $i\in\{ 0,1,\cdots,m\}$ must holds one of its favorite objects.
    Form the construction of the preference profile, we know that $A_i^1$ and $A_i^2$ must hold their favorite objects respectively for every $i$.
    If $\sigma(A_i^3)$ is not one of the favorite objects of $A_i^3$, then one of $\sigma(A_i^2)$ and $\sigma(A_i^1)$ is one of the favorite objects of $A_i^3$ and it can be swapped to $A_i^3$ directly.
    Let $X$ be one of $\overline{X}_i^{n_i}$, $X_i^{q_i}$ and $X_i^{p_i}$ such that $\sigma(X)$ is not one of the favorite objects of $X$ and let $Y$ be a neighbor of $X$ in $\{\overline{X}_i^{n_i},X_i^{q_i},X_i^{p_i}\}$.
    Then, we have that $\sigma(X)\succeq_Y \sigma(Y)$ and $\sigma(Y)\succeq_X \sigma(X)$. We make a swap between $X$ and $Y$ now.
    So, we can always get a reachable assignment $\sigma'$ such that each agent holds one of its favorite objects in $\sigma'$ resulting from $\sigma$.
\qed
\end{proof}

By this lemma, we know that if $I_{WOR}$ is a yes-instance, then in any Pareto optimal assignment for $I_{WOR}$ every agent holds one of its most favorite objects.
On the other hand, if there is a Pareto optimal assignment where every agent holds one of its most favorite objects, then $C_m$ must hold $t$.
Thus the above reduction also reduces \textsc{2P1N-SAT} to \textsc{Weak Pareto Efficiency} in paths.
The instance $I_{WOR}$ has a Pareto optimal assignment such that every agent holds one of its most favorite objects if and only if the \textsc{2P1N-SAT} instance $I_{SAT}$ is a yes-instance.

\begin{theorem}
  \label{the_pareto}
  \textsc{Weak Pareto Efficiency} is NP-hard even when the network is a path.
\end{theorem}

Next, we discuss the problem of maximizing the utilitarian social welfare.
In this problem, we will associate a value function $f_i(\cdot)$ on the objects for each agent $i$.
For any two objects $o_a$ and $o_b$, the value function holds that $f_i(o_a)> f_i(o_b)$ if and only if $o_{a}\succ_i o_{b}$ (resp., $f_i(o_a)= f_i(o_b)$ if and only if $o_{a}=_i o_{b}$).
Thus, each agent has the maximum value on the most favorite objects.
By Lemma~\ref{newadded1}, in the above instance $I_{WOR}$, all agents get their maximum value, i.e., the utilitarian social welfare is maximum, if and only if the \textsc{2P1N-SAT} instance $I_{SAT}$ is a yes-instance.

\begin{theorem}
  \label{the_sw}
    Finding a reachable assignment that maximizes the utilitarian social welfare is NP-hard even when the network is a path.
\end{theorem}

\section{Weak Preference Version in Stars}\label{sec-star}
In this section, we assume that the input graph is a star and that the preferences allow ties.
We will consider the reachability of an object for an agent and the problem of maximizing the utilitarian social welfare.

\subsection{\textsc{Weak Object Reachability} in Stars}
We have shown that \textsc{Weak Object Reachability} in paths is NP-hard.
Next, we show that  \textsc{Weak Object Reachability} in stars can be solved in polynomial time.
We first show several properties of the problem, which will allow us to transform the original problem to the problem of whether the target object is reachable for the center agent under some constraints. We further show that if such kinds of feasible assignments exist, then there is a sequence of swaps to reach a feasible assignment such that each leaf agent participates in at most one swap, which is called a \emph{simple} sequence. Inspired by the idea of the algorithm for the strict version by \citet{DBLP:conf/ijcai/GourvesLW17},
we construct an auxiliary graph and show that a satisfying simple sequence exists if and only if the auxiliary graph has a directed path from one given vertex to another given vertex. Thus,
we can solve the problem by finding a directed path in a graph.

Without loss of generality, we assume that the network is a star of $n$ vertices and the center agent of the star is $n$;
the initially endowed object assigned to each agent $i$ is $o_i$.
It is trivial for the case of asking whether the center object $o_n$ is reachable for a leaf agent $k$, where we only need to check whether $o_n \succeq_k o_k $ and $o_k \succeq_n o_n$.
Next, we consider whether an object of a leaf agent is reachable for another agent.
Without loss of generality, we assume that the problem is to ask whether object $o_1$ is reachable for agent $k$, where $k$ is allowed to be $n$.

Since the network is a star, each feasible swap must include the center agent $n$. This property will be frequently used in our analysis.
Next, we show more properties of the problem.

\begin{lemma}\label{lem_star_delete}
    For any agent $i_0\not\in\{1,k,n\}$, if $o_{i_0}\succ_n o_1$ or $o_n \succ_n o_{i_0}$, then deleting agent $i_0$ and its initially endowed object $o_{i_0}$ will not change the reachability of the instance.
\end{lemma}
\begin{proof}
For any agent $i_0$ with $o_n \succ_n o_{i_0}$, it cannot participate in a feasible swap, because the center agent $n$ will not hold an object worse than its initially endowed object $o_n$.
So we can simply delete it. Next, we assume that $o_{i_0}\succ_n o_1$.
Agent $i_0$ can only be a leaf agent because $i_0\neq n$. Before agent $k$ holds object $o_1$, the center agent $n$ cannot hold object $i_0$, otherwise agent $n$ will not
get $o_1$ anymore and then agent $k$ cannot get $o_1$. Since $i_0\neq k$, we can assume that agent $i_0$ does not participate in any swap and delete it.
\qed
\end{proof}

\begin{lemma}\label{lem_star_case1}
If $o_k\succ_k o_1$ or $o_n\succ_n o_1$, then $o_1$ is not reachable for agent $k$.
\end{lemma}
\begin{proof}
Since any agent will not receive an object worse than its initially endowed object, we know that agent $k$ will not receive $o_1$ if $o_k\succ_k o_1$.
For the case of $o_n\succ_n o_1$, the center agent $n$ will not receive $o_1$ and then $k$ will not be able to get object $o_1$.
\qed
\end{proof}

Based on Lemma~\ref{lem_star_delete} and Lemma~\ref{lem_star_case1}, we can further assume that $o_1\succeq_n o_i \succeq_n o_n$ for any $i\neq k$ and $o_1\succeq_k o_k$.

\begin{lemma}\label{lem_1k}
    If $o_1\succ_n o_k$ and $k\neq n$, then $o_1$ is not reachable for agent $k$.
\end{lemma}
\begin{proof}

    Assume to the contrary that that $o_1$ is reachable for agent $k$. In the sequence of assignments from the initial one to a reachable one via swaps,
    we let $o_{x_1}, o_{x_2},\dots, o_{x_l}$ be the sequence of different objects assigned to agent $k$,
    where $o_{x_1}=o_k$ and $o_{x_l}=o_1$. For each $1\leq i \leq l-1$, there is a swap between agent $n$ and agent $k$ where  agent $n$ holds object $o_{x_{i+1}}$
    and agent $k$ holds object $o_{x_{i}}$ (before the swap). This means that $o_{x_{i}} \succeq_n o_{x_{i+1}}$. So we get
    $$o_{x_1} \succeq_n o_{x_2} \succeq_n\dots \succeq_n o_{x_l}.$$
    This is in contradiction with the fact that $o_1\succ_n o_k$. So we know that $o_1$ is not reachable for agent $k$ for this case.
\qed
\end{proof}

To solve \textsc{Weak Object Reachability} in stars, we consider whether $k=n$ or not. If $k=n$, it is to check whether object $o_1$ is reachable for the center agent $n$. 
If $k\neq n$, we will transfer it to checking the reachability of $o_1$ for the center agent $n$ under some constraints. Our idea is based on the following observations. Before object $o_1$ reaches agent $k\neq n$, object $o_1$ must be held by the center agent $n$. We will first consider the case that $k\neq n$.

\begin{definition}
An assignment $\sigma$ is called \emph{crucial} if it holds that $\sigma(n)=o_1$, $o_1\succeq_k \sigma(k)$ and $\sigma(k)\succeq_n o_1$.
\end{definition}
\begin{ob}\label{ob-equel}
Assume that $k\neq n$.
Object $o_1$ is reachable for agent $k$ if and only if there is a reachable crucial assignment.
\end{ob}
\begin{proof}
Before the last swap to reach a reachable assignment $\sigma$ such that $\sigma(k)=o_1$, the assignment is a reachable crucial assignment.
Furthermore, for a crucial assignment, we can make a trade between $n$ and $k$ so that agent $k$ can get object $o_1$.
\qed
\end{proof}

Next, we will analyze the properties of reachable crucial assignments and design an algorithm to find them if they exist.

\begin{lemma}\label{lem_sequence}
Assume that $k\neq n$.
    For any sequence of swaps $(\sigma_0,\sigma_1,\dots,\sigma_t)$ such that $\sigma_t$ is a reachable crucial assignment (if it exists),
    it holds that $\sigma_i(k)\succeq_n o_1$ for any $i\in \{0,1,\dots, t\}$.
\end{lemma}
\begin{proof}
Assume to the contrary that there exists $\sigma_i$ $(0\leq i \leq t)$ such that  $o_1 \succ_n \sigma_i(k)$. We take $\sigma_i$ as the initial endowment and then
$o_1$ is not reachable for agent $k$ by Lemma~\ref{lem_1k} and Observation~\ref{ob-equel}. This is in contradiction with the existence of $\sigma_t$. So the lemma holds.
\qed
\end{proof}

\begin{definition}
A sequence of swaps is called \emph{simple} if each leaf agent in the star participates in at most one swap.
\end{definition}

\begin{lemma}\label{lem_sequence2}
If there exists a reachable crucial assignment, then there exists
a simple sequence of swaps $(\sigma_0,\sigma_1,\dots,\sigma_t)$ such that $\sigma_t$ is crucial and $\sigma_i(k)\succeq_n o_1$ holds for any $i\in \{0,1,\dots, t\}$.
\end{lemma}
\begin{proof}
Assume reachable crucial assignments exist. We let  $\psi=(\sigma_0,\sigma_1,\dots,\sigma_t)$ be one of the shortest sequences of swaps such that $\sigma_t$ is a reachable crucial assignment. Assume to the contrary that at least one leaf agent participates in two swaps in $\psi$. Then $\sigma_i(n)=\sigma_j(n)$ will hold for two different indices $0\leq i < j \leq t$. We assume that $j$ is the maximum index satisfying the above condition. By the selection of the index $j$, we know that any leaf agent that participates in a swap in $(\sigma_j,\sigma_{j+1},\dots,\sigma_t)$ will not participate in any swap before $\sigma_j$. Thus, after deleting the sequence of swaps in $(\sigma_i,\sigma_{i+1},\dots,\sigma_j)$
    from $\psi$, we will still get a feasible sequence of swaps. However, the length of the new sequence of swaps is shorter than that of $\psi$, a contradiction to the choice of $\psi$.
    So we know that $\psi$ must be simple. By Lemma~\ref{lem_sequence}, we know that $\sigma_i(k)\succeq_n o_1$ holds for any $i\in \{0,1,\dots, t\}$. Thus, the lemma holds.
\qed
\end{proof}


\begin{lemma}\label{lem_star-algorithm}
    Assume that $k\neq n$, $o_1 \succeq_k o_k$ and $o_k \succeq_n o_1$.
    There is an algorithm that runs in $O(n^2)$ time to check the existence of a simple sequence of swaps $(\sigma_0,\sigma_1,\dots,\sigma_t)$ such that $\sigma_t$ is crucial and $\sigma_i(k)\succeq_n o_1$ holds for any $i\in \{0,1,\dots, t\}$, and to find one if it exists.
\end{lemma}

\begin{proof}
We will transform the problem to a problem of finding a simple directed path in an auxiliary graph following the idea of the algorithm for strict preferences given in~\citep{DBLP:conf/ijcai/GourvesLW17}.
For a simple sequence of swaps, the center agent and a leaf agent make a trade only when the leaf agent holds its initially endowed object.
We construct an auxiliary directed graph $G_D=(N,E_D)$ on the set $N$ of agents.
There is an arc $\overrightarrow{ij}$ from  agent $i\in N$ to agent $j\in N\setminus\{n,k,i\}$ if and only if
the center agent $n$ and the leaf agent $j$ can rationally trade when the center agent $n$ holds object $o_i$ and the leaf agent $j$ holds object $o_j$,
i.e., $o_i\succeq_j o_j$ and $o_j\succeq_n o_i$.
There is an arc $\overrightarrow{ik}$ from  agent $i\in N\setminus\{1,k\}$ to agent $k$ if and only if
$o_1\succeq_k o_i\succeq_k o_k$ and $o_k\succeq_n o_i \succeq_n o_1$.
The auxiliary graph can be constructed in $O(n^2)$ by using some data structure.
Next, we show that a simple sequence of swaps $(\sigma_0,\sigma_1,\dots,\sigma_t)$ such that $\sigma_t$ is crucial and $\sigma_i(k)\succeq_n o_1$ holds for any $i\in \{0,1,\dots, t\}$ exists if
and only if there is a simple directed path from $n$ to 1 in graph $G_D$.

It is easy to see that such a simple sequence of swaps implies a simple directed path from $n$ to 1 in $G_D$.
We consider the other direction and assume that a simple directed path $P$ from $n$ to 1 in $G_D$ exists. We can further assume that no vertex appears more than once in $P$ since a simple path always exists if a path between two vertices exists. We can get a sequence $(\sigma_0,\sigma_1,\dots,\sigma_t)$ of swaps according to $P$:
swap $(\sigma_i,\sigma_{i+1})$ is a trade between $n$ and the head of the $i$th arc in $P$.
First, the sequence of swaps is simple because no vertex appears more than once in $P$.
Second, it holds that $\sigma_t(n)=o_1$ because the last swap happens between agent $n$ and agent $1$ and this is the first swap including agent 1.
Third, it holds that $o_1\succeq_k \sigma_i(k)$ for every $i$ because agent $k$ has an initially endowed object $o_k$ such that $o_1 \succeq_k o_k$ and $G_D$ has no arc from $j$ to $k$ for any $j$ with $o_j \succ_k o_1$.
Last, we have $\sigma_i(k)\succeq_n o_1$ for any $i\in \{0,1,\dots, t\}$ because agent $k$ has an initially endowed object $o_k$ such that  $o_k \succeq_n o_1$ and $G_D$ has no arc from $i$ to $k$ for any $i$ with $o_1 \succ_n o_i$.
Therefore, we know that $\sigma_t$ is crucial and $\sigma_i(k)\succeq_n o_1$ holds for any $i\in \{0,1,\dots, t\}$.

A simple directed path from $n$ to 1 in graph $G_D$ can be computed in linear time by the depth-first search. The lemma holds.
\qed
\end{proof}

For the case of $k=n$, we still need the following Lemma~\ref{lem_sequencen} and Lemma~\ref{lem_star-algorithmn}. The proof of Lemma~\ref{lem_sequencen} is referred to the first part of the proof of Lemma~\ref{lem_sequence2}. We omit the redundant arguments.
The same arguments in Lemma~\ref{lem_star-algorithm} can be used to prove Lemma~\ref{lem_star-algorithmn}.
However, the construction of the auxiliary graph is a little bit different.

\begin{lemma}\label{lem_sequencen}
If object $o_1$ is reachable for the center agent $n$, then there exists
a simple sequence of swaps $(\sigma_0,\sigma_1,\dots,\sigma_t)$ such that $\sigma_t(n)=o_1$.
\end{lemma}


    \begin{lemma}\label{lem_star-algorithmn}
        Verifying the existence of a simple sequence of swaps $(\sigma_0,\sigma_1,\dots,\sigma_t)$ such that $\sigma_t(n)=o_1$, and constructing it if it exists, can be done in $O(n^2)$.
    \end{lemma}

\begin{proof}
We construct an auxiliary directed graph $G_D=(N,E_D)$ on the set $N$ of agents.
There is an arc $\overrightarrow{ij}$ from  agent $i\in N$ to agent $j\in N\setminus\{n,i\}$ if and only if
the center agent $n$ and the leaf agent $j$ can rationally trade when the center agent $n$ holds object $o_i$ and the leaf agent $j$ holds object $o_j$,
i.e., $o_i\succeq_j o_j$ and $o_j\succeq_n o_i$.
Similar to the argument in the proof of Lemma~\ref{lem_star-algorithm}, we can prove that
a simple sequence of swaps $(\sigma_0,\sigma_1,\dots,\sigma_t)$ such that $\sigma_t(n)=o_1$ exists if
and only if there is a simple directed path from $n$ to 1 in graph $G_D$.
\qed
\end{proof}

The main steps of our algorithm are listed in Algorithm~\ref{alg:star}.

\begin{algorithm}[h]
    \caption{The Algorithm for Weak Object Reachability in Stars}
    \label{alg:star}
    \KwIn{An instance in a star, where we assume that $n$ is the center agent in the star, $o_1\neq o_n$ and $k\neq 1$.}
    \KwOut{To determine whether $o_1$ is reachable for agent $k$}
    \For{each $i\in \{2, 3, \dots, n-1\}\setminus\{k\}$}
    {\If {$o_{i}\succ_n o_1$ or $o_n \succ_n o_{i}$} {delete agent $i$ and object $o_i$\;} }
    \If {$o_k\succ_k o_1$ or $o_n\succ_n o_1$} {\textbf{return no} and \textbf{quit}\;}
    \If {$k\neq n$ and $o_1\succ_n o_k$} {\textbf{return no} and \textbf{quit}\;}
    \If {$k\neq n$} {Check whether there is a simple directed path from $n$ to $1$ in the graph $G_D$ constructed in the proof of Lemma~\ref{lem_star-algorithm}\;
     \textbf{return yes} if it exists and \textbf{no} otherwise\;}
    \Else {Check whether there is a simple directed path from $n$ to $1$ in the graph $G_D$ constructed in the proof of Lemma~\ref{lem_star-algorithmn}\;
     \textbf{return yes} if it exists and \textbf{no} otherwise.}
\end{algorithm}

\begin{theorem}
  \label{lem_wor}
  \textsc{Weak Object Reachability} in stars can be solved in $O(n^2)$ time.
\end{theorem}
\begin{proof}
We prove this theorem by proving the correctness of Algorithm~\ref{alg:star}. By Lemma~\ref{lem_star_delete}, we can do Steps 1-3 to simplify the instance.
By Lemma~\ref{lem_star_case1}, we can verify the correctness of Steps 4-5, and  by Lemma~\ref{lem_1k},
we can verify the correctness of Steps 6-7.
After Step 7, it always holds that $o_{1}\succeq_k o_k$, $o_1 \succeq_n o_{n}$ and $o_k \succeq_n o_{1}$ for $k\neq n$.
For the case of $k\neq n$, the correctness of Steps 8-10 is derived from Observation~\ref{ob-equel}, Lemma~\ref{lem_sequence2} and Lemma~\ref{lem_star-algorithm}.
For the case of $k= n$, it is to check whether $o_1$ is reachable for the center agent $n$ and the correctness of Steps 12-13 is derived from Lemma~\ref{lem_sequencen} and Lemma~\ref{lem_star-algorithmn}.

Before Step 8, the algorithm uses linear time. Steps 8-10 use $O(n^2)$ time and Steps 12-13 use $O(n^2)$ time.
\qed
\end{proof}

\subsection{Maximizing The Utilitarian Social Welfare in Stars}

For \textsc{Weak Pareto Efficiency} in stars, we did not find a polynomial-time algorithm or a proof for the NP-hardness.
However, we can prove the hardness of the problem to find a solution maximizing the utilitarian social welfare.
In this problem, the preference profile is replaced with the value function.

To prove the NP-hardness, we give a reduction from the known NP-complete problem -- the \textsc{Directed Hamiltonian Path} problem~\cite{DBLP:books/fm/GareyJ79}, which is to find a directed path visiting each vertex exactly once in a directed graph $D=(V,A)$.
We construct an instance $I$ of the problem such that $I$ admits a reachable assignment with the utilitarian social welfare at least $3|V|+|A|-1$ if and only if the graph $D$ has a directed Hamiltonian path starting from a given vertex $s\in V$. We simply assume that $s$ has no incoming arcs.

Instance $I$ is constructed as follows. It contains $|V|+|A|+1$ agents and $|V|+|A|+1$ objects.
For each arc $e\in A$, there is an associated agent $a_e$, called an \emph{arc agent}.
For each vertex $v\in V$, there is an associated agent $a_v$, called a \emph{vertex agent}.
So there are $|A|$ arc agents and $|V|$ vertex agents. There is also a \emph{center agent} $a_c$, which is adjacent to all arc agents
and vertex agents to form a star.

In the initial endowment $\sigma_0$,
the object assigned to the center agent $a_c$ is $o_c$, called the \emph{center object}; for each vertex agent $a_v$, the object assigned to it is
$o_v$, called a \emph{vertex object}; and for each arc agent $a_e$,
the object assigned to it is $o_e$, called an \emph{arc object}.

Next, we construct the value function for each agent.
For the center agent $a_c$, all objects have the same value of $0$.
For the starting vertex agent $a_s$, the initially endowed object $o_s$ is valued 1, the center object $o_c$ is valued 2, and all other objects are valued 0.
For  any vertex agent $a_v$ with $v\in V\setminus\{s\}$, the initially endowed object $o_v$ is valued 1, arc object $o_e$
with $e$ being an arc from a vertex to the vertex $v$
is valued 2, and all other objects are valued 0.
For any arc agent $a_e$ with $e\in A$, where $e=\overrightarrow{uv}$ is an arc starting from $u$ to $v$, the initially endowed object $o_e$ is valued 1, the vertex object $o_u$ is valued 2, and all other objects are valued 0.

See Figure~\ref{fig10} and Table~\ref{tableexmp} for an example to construct the instance.

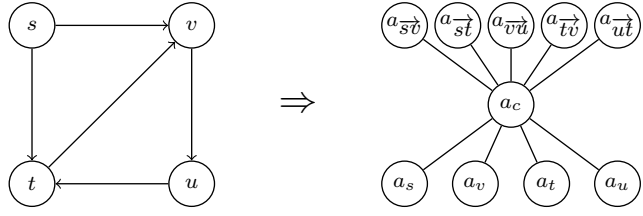
\begin{figure}[htbp]
    \centering
    \begin{tikzpicture}
        [scale = 0.7, line width = 0.5pt,solid/.style = {circle, draw, fill = black, minimum size = 0.3cm},empty/.style = {circle, draw, fill = white, minimum size = 0.6cm}]
        \node [empty, label = center:$s$] (S) at (-9,1.5) {};
        \node [empty, label = center:$t$] (T) at (-9,-1.5) {};
        \node [empty, label = center:$v$] (V) at (-6,1.5) {};
        \node [empty, label = center:$u$] (U) at (-6,-1.5) {};

        \draw[->] (S) -- (V);
        \draw[->] (S) -- (T);
        \draw[->] (V) -- (U);
        \draw[->] (T) -- (V);
        \draw[->] (U) -- (T);

        \node [label = center:\Large$\Rightarrow$] (R) at (-4,0) {};

        \node [empty, label = center:$a_{\overrightarrow{sv}}$] (A) at (-2,1.5) {};
        \node [empty, label = center:$a_{\overrightarrow{st}}$] (B) at (-1,1.5) {};
        \node [empty, label = center:$a_{\overrightarrow{vu}}$] (C) at (0,1.5) {};
        \node [empty, label = center:$a_{\overrightarrow{tv}}$] (D) at (1,1.5) {};
        \node [empty, label = center:$a_{\overrightarrow{ut}}$] (E) at (2,1.5) {};

        \node [empty, label = center:$a_c$] (F) at (0,0) {};

        \node [empty, label = center:$a_s$] (G) at (-2,-1.5) {};
        \node [empty, label = center:$a_v$] (H) at (-0.67,-1.5) {};
        \node [empty, label = center:$a_t$] (I) at (0.67,-1.5) {};
        \node [empty, label = center:$a_u$] (J) at (2,-1.5) {};

        \draw (A) -- (F);

        \draw (B) -- (F);
        \draw (C) -- (F);
        \draw (D) -- (F);
        \draw (E) -- (F);
        \draw (G) -- (F);
        \draw (H) -- (F);
        \draw (I) -- (F);
        \draw (J) -- (F);
    \end{tikzpicture}
    \caption{An example of the construction in the proof of Theorem~\ref{the_sv_star}} \label{fig10}
\end{figure}

\begin{table}[h]
    \centering
    \caption{The value functions of the example in Figure 10}\label{tableexmp}
    \begin{tabular}{c|p{0.5cm}<{\centering}|p{0.5cm}<{\centering}|p{0.5cm}<{\centering}|p{0.5cm}<{\centering}|p{0.5cm}<{\centering}|p{0.5cm}<{\centering}|p{0.5cm}<{\centering}|p{0.5cm}<{\centering}|p{0.5cm}<{\centering}|p{0.5cm}<{\centering}}
        \hline
        \diagbox{Objects}{Agents} & $a_c$ & $a_s$ & $a_v$ & $a_t$ & $a_u$ & $a_{\overrightarrow{sv}}$ & $a_{\overrightarrow{st}}$ & $a_{\overrightarrow{vu}}$ & $a_{\overrightarrow{tv}}$ & $a_{\overrightarrow{ut}}$ \\
        \hline
        \rule{0pt}{8pt}$o_c$ &0&2&0&0&0&0&0&0&0&0\\
        \hline
        \rule{0pt}{8pt}$o_s$ &0&1&0&0&0&2&2&0&0&0\\
        \hline
        \rule{0pt}{8pt}$o_v$ &0&0&1&0&0&0&0&2&0&0\\
        \hline
        \rule{0pt}{8pt}$o_t$ &0&0&0&1&0&0&0&0&2&0\\
        \hline
        \rule{0pt}{8pt}$o_u$ &0&0&0&0&1&0&0&0&0&2\\
        \hline
        \rule{0pt}{8pt}$o_{\overrightarrow{sv}}$ &0&0&2&0&0&1&0&0&0&0\\
        \hline
        \rule{0pt}{8pt}$o_{\overrightarrow{st}}$ &0&0&0&2&0&0&1&0&0&0\\
        \hline
        \rule{0pt}{8pt}$o_{\overrightarrow{vu}}$ &0&0&0&0&2&0&0&1&0&0\\
        \hline
        \rule{0pt}{8pt}$o_{\overrightarrow{tv}}$ &0&0&2&0&0&0&0&0&1&0\\
        \hline
        \rule{0pt}{8pt}$o_{\overrightarrow{ut}}$ &0&0&0&2&0&0&0&0&0&1\\
        \hline
    \end{tabular}
\end{table}

From the value functions of instance $I$, we can see the following properties.

\begin{enumerate}
\item [(P1)] the utilitarian social welfare of the initial endowment is exactly $|V|+|A|$;
\item [(P2)] the vertex agent $a_s$ can only participate in a trade with $a_c$ where $o_s$ and $o_c$ are swapped;
\item [(P3)] the first swap must happen between agents $a_s$ and $a_c$, because all other agents value object $o_c$ as 0;
\item [(P4)] for a vertex object $o_v$, among all vertex agents, only vertex agent $a_v$ values it as 1 and all other vertex agents value it as 0;
\item [(P5)] for a vertex object $o_v$, among all arc agents, only arc agents whose arc starts from vertex $v$ value it as 2, and all other arc agents value it as 0;
\item [(P6)] for an arc object $o_e$, among all vertex agents, only vertex agents which is the endpoint of edge $e$ value it as 2, and all other vertex agents value it as 0;
\item [(P7)] for an arc object $o_e$, among all arc agents, only arc agent $a_e$ values it as 1 and all other arc agents value it as 0.
\end{enumerate}

From (P4), (P5), (P6) and (P7), we can see that
\begin{enumerate}
\item [(P8)] when a vertex object $o_v$ is held by the center agent $a_c$, the next swap can only be a trade
between the center agent $a_c$ and an arc agent $a_e$ with $e$ being an arc starting from vertex $v$;
\item [(P9)] when an arc object $o_e$ is held by the center agent $a_c$, the next swap can only be a trade
between the center agent $a_c$ and the vertex agent $a_v$ with $v$ being the ending point of arc $e$.
\end{enumerate}

Based on the above properties, we prove the following results.

\begin{lemma} \label{resultf}
    Instance $I$ has a reachable assignment such that the total value of all agents is at least $3|V|+|A|-1$ if and only if there is a directed Hamiltonian path starting from vertex $s$ in the directed graph $D$.
\end{lemma}

\begin{proof}

    Assume that there is a directed Hamiltonian path $v_1v_2v_3\dots,v_{|V|}$ starting from $s$ in the directed graph $D$, where $v_1=s$
    and the arc from $v_i$ to $v_{i+1}$ is denoted by $e_{i}$.
    We show that there is a satisfying sequence of swaps in $I$.
    Note that each swap in a star happens between the center agent and a leaf agent. So a sequence of swaps can
    be specified by a sequence of leaf agents. So we will use a sequence of leaf agents to denote a sequence of swaps.
    In fact, we will show that $(a_{v_1}, a_{e_1}, a_{v_2}, a_{e_2},\dots, a_{e_{|V|-1}}, a_{v_{|V|}})$ is a satisfying sequence of swaps in $I$.

    It is easy to verify that the first two swaps can be executed. We can prove by induction that, for each $i>1$, before the swap between $a_c$ and $a_{v_i}$ (resp., $a_{e_i}$), agent $a_c$ holds object $o_{e_{i-1}}$ (resp., $o_{v_{i-1}}$)  and  agent $a_{v_i}$ holds object $o_{v_i}$ (resp., agent $a_{e_i}$ holds object $o_{e_i}$).
    According to (P4), (P5), (P6), (P7) and the fact that the center agent $a_c$ has no difference among all the objects, we know that the sequence of swaps
     $(a_{v_1}, a_{e_1}, a_{v_2}, a_{e_2},\dots, a_{e_{|V|-1}}, a_{v_{|V|}})$ can be executed. Furthermore, for each swap in the sequence, the leaf agent participating in the swap will increase its value by 1. There are $2|V|-1$ swaps and the initial endowment has a value of $|V|+|A|$. So in the final assignment, the total value is $|V|+|A|+(2|V|-1)=3|V|+|A|-1$.

Next, we assume that there is a reachable assignment $\sigma$ with the total value of agents at least $3|V|+|A|-1$ in $I$.
Based on this assumption, we construct a Hamiltonian path starting from $s$ in $D$. We let a sequence of leaf agents $\pi=(a_{x_1},a_{x_2}, \dots, a_{x_l})$ to denote the sequence of swaps from the initial endowment to the final assignment $\sigma$, where we assume that $x_i\neq x_{i+1}$ for all $i$ since two consecutive swaps happened between the same pair of agents mean doing nothing and they can be deleted.
We will show that $\pi$ is a sequence of alternating elements between vertices and arcs in $D$.


 By (P3), we know that $x_1$ is $s$. We consider two cases: all vertex agents in $\pi$ are different or not.
 We first consider the case that all vertex agents in $\pi$ are different. For this case, it is easy to see that all arc
 agents in $\pi$ are also different by the value functions of the agents.
 Now all agents in $\pi$ are different. By (P8) and (P9), we know that the sequence of vertices and arcs in $D$ corresponding to $\pi$ is a directed path starting from $s$.

 We also have the following claims by the fact that each vertex and arc agent appears at most once in $\pi$. For each swap in $\pi$, if it is a trade between $a_c$ and a vertex agent $a_v$,
then the object hold by agent $a_v$ changes from the initially endowed object $o_v$ to an arc object $o_e$, where agent $a_v$ must value object $o_e$ as 2.
Thus, the value of agent $a_v$ increases by 1.
For each swap in $\pi$, if it is a trade between $a_c$ and an arc agent $a_e$,
then the object hold by agent $a_e$ changes from the initially endowed object $o_e$ to a vertex object $o_v$, where agent $a_e$ must value object $o_v$ as 2.
Thus, the value of agent $a_e$ increases by 1. In any case, each swap will increase the total value by exactly 1.
By (P1), the utilitarian social welfare of the initial endowment is $|V|+|A|$. So there are exactly $3|V|+|A|-1-(|V|+|A|)=2|V|-1$
swaps in $\pi$. The length of $\pi$ is  $2|V|-1$ and then $\pi$ contains $|V|$ vertices and $|V|-1$ arcs, where all vertices are different. This can only be a Hamiltonian path.

Next, we consider the case that some vertex agents appear at least twice in $\pi$.
 Let $\pi'$ be the shortest subsequence of $\pi$ starting from the beginning agent $a_s$ such that a vertex agent appears twice in $\pi'$. Let $a_v$ be the last agent in $\pi'$.
 By the choice of $\pi'$, we know that
  each agent in $\pi'$ except the last one appears at most once in $\pi'$. By (P8) and (P9), we know that $\pi'$ is corresponding to a directed path starting from $s$ in $D$, where only the last vertex $v$ in the path appears twice.  By (P2), we know that $v$ cannot be $s$. Thus, there is
  an arc agent $a_{e_1}$ before the first appearance of $a_v$ in $\pi'$. Let $a_{e_2}$ be the agent before the second appearance of $a_v$ in $\pi'$, i.e., the last but one agent in $\pi'$. We can see that before the second swap between $a_c$ and $a_v$, center agent $a_c$ holds object $o_{e_2}$ and vertex agent $a_v$ holds object $o_{e_1}$,
  where both $e_1$ and $e_2$ are arcs with the ending point being $v$.
   After this swap, center
  agent $a_c$ will hold object $o_{e_1}$ and vertex agent $a_v$ will hold object $o_{e_2}$. From then on, only vertex agent $a_{v}$
  can make a trade with center agent $a_c$. We have assumed that there are no two consecutive swaps that happened between the same pair of agents. Thus the second $a_v$ must be the last agent in $\pi$. Furthermore, in the last swap between $a_c$ and $a_v$, the value of
  $a_v$ will not increase. By the above analysis for the first case,
  we know that $\pi''=(a_{x_1},a_{x_2}, \dots, a_{x_{l-2}})$ is corresponding to a simple directed path, where $a_{x_{l-2}}$ is a vertex agent. In $\pi$, the total value of the agents
  can increase by at most 1 during the last but one swap between $a_c$ and $a_{x_{l-1}}$, and can increase by 0 during the last swap between $a_c$ and $a_{x_{l}}$. There is an odd number of agents in $\pi''$ since vertex agents and arc agents appear alternately in it and both the first and the last agents are vertex agents. Each swap in $\pi''$ can increase the total value by 1. Since the total value of the last assignment is at least $3|V|+|A|-1$ and the initial value is $|V|+|A|$, we know that the length of $\pi''$ is at least $2|V|-1$, which implies that $\pi''$ is corresponding to a Hamiltonian path.
    \qed
\end{proof}
Lemma~\ref{resultf} implies that

\begin{theorem}
    \label{the_sv_star}
    Finding a reachable assignment that maximizes the utilitarian social welfare is NP-hard even when the network is a star.
\end{theorem}

\section{Conclusion and Discussion}\label{sec-con}
In this paper, we investigate some problems about how to obtain certain assignments from an initial endowment via a sequence of rational trades between two agents under two different preferences: ties are allowed (weak preferences) or not (strict preferences).
We assume that the agents are embedded in a social network and only neighbors are allowed to exchange their objects. This is a reasonable assumption since it well simulates the situation in real life. To understand the computational complexity of the problem, we consider different network structures. See Table~\ref{table1} in the introduction for a survey of the results.

For \textsc{Object Reachability}, now we know it can be solved in polynomial time in paths. Recently, it was shown in \cite{DBLP:journals/corr/abs-1905-04219} that it can also be solved in polynomial time in cycles, which implies the problem is polynomial-time solvable in graphs with maximum degree 2.
Will the problem become NP-hard in graphs with maximum degree 3? How about the computational complexity of the problem in bounded-degree trees?
Finding more social structures under which the problem is polynomial-time solvable will be helpful for us to understand the nature of the problem.

In this paper, we regard the preferences of agents as independent. In real life, the preferences of agents may be similar, i.e., some objects are preferred by most agents.
Will the problem become easier if the preferences have some properties like single peakedness or single crossingness?

In our model, we only consider trades between two agents. Trades among three or more agents also exist, such as in kidney exchange, a trade among three agents is common and
a trade among nine agents (a 9-cycle kidney
exchange) was performed successfully in 2015 at two San Francisco hospitals~\cite{DBLP:journals/jco/Mak-Hau17}.
General trades among an arbitrary number of agents have been considered in \cite{DBLP:conf/atal/DamammeBCM15}. It would be
 interesting to further reveal more properties of the model under general trades.
In kidney exchange, each agent is allowed to participate in at most one trade. In our model, an agent can participate in an arbitrary number of trades.
Another direction for further study is to consider the model with an upper bound on the number of trades for each agent.
It has been proved in~\cite{DBLP:conf/sagt/SaffidineW18} that \textsc{Object Reachability} is NP-hard in graphs with maximum degree 4 even if the maximum number of swaps that each agent can participate is 2.

We believe it is worthy to further study more object allocation models under exchange operations, especially for problems coming from real-life applications.

\section*{Acknowledgements}
This work was supported by the National Natural Science Foundation of China,
under grants 61972070 and 61772115.

\end{document}